\documentclass[a4paper,onecolumn,accepted=2022-10-05]{quantumarticle}
\pdfoutput=1
\usepackage[utf8]{inputenc}
\usepackage[english]{babel}
\usepackage[T1]{fontenc}
\usepackage[]{hyperref}
\usepackage{doi}

\usepackage{svg}

\usepackage{tikz}

\usepackage{float}

\graphicspath{{figures/}} 

\usepackage{bbold}
\usepackage{amsmath}
\usepackage{mathrsfs}
\usepackage{amsthm}
\usepackage{amsfonts,amssymb,amsthm, bbm,braket}
\usepackage{graphicx}   
\usepackage{subcaption}
\usepackage[export]{adjustbox}[2011/08/13]
\usepackage{amsbsy} 
\usepackage[bold]{hhtensor} 
\usepackage{algorithm}
\usepackage[noend]{algpseudocode}

\usepackage{multirow}
\usepackage{tcolorbox}

\newcommand{\bigo}{\mathcal{O}}
\newcommand{\C}{\mathcal{C}}

\newcommand{\T}{\mathcal{T}}
\newcommand{\CP}{\mathcal{CP}}
\newcommand{\TP}{\mathcal{TP}}
\newcommand{\CPTP}{\mathcal{CPTP}}

\newcommand{\Id}{\mathcal{I}_{d}}

\newcommand{\va}{\textbf{a}}
\newcommand{\vb}{\textbf{b}}
\newcommand{\vo}{\textbf{o}}

\newcommand{\vs}{\textbf{s}}

\newcommand{\vp}{\textbf{p}}
\newcommand{\vq}{\textbf{q}}
\newcommand{\Tr}{\mbox{Tr}}

\newcommand{\Pls}{\hat{\Phi}_{LS}}

\newcommand{\LSE}{\hat{\Phi}_{LS}}
\newcommand{\PLS}{\hat{\Phi}_{PLS}}
\newcommand{\CPone}{\hat{\Phi}_{CP1}}

\DeclareMathOperator*{\argmin}{argmin}
\DeclareMathOperator*{\argmax}{argmax}
\DeclareMathOperator*{\Herm}{Herm}

\renewcommand\bra[1]{{\langle{#1}|}}
\renewcommand\ket[1]{{|{#1}\rangle}}

\newtheorem*{rep@theorem}{\rep@title}
\newcommand{\newreptheorem}[2]{%
\newenvironment{rep#1}[1]{%
 \def\rep@title{#2 \ref{##1}}%
 \begin{rep@theorem}}%
 {\end{rep@theorem}}}

\newtheorem{definition}{Definition}
\newtheorem{corollary}{Corollary}
\newtheorem{theorem}{Theorem}
\newtheorem{lemma}{Lemma}
\newreptheorem{lemma}{Lemma}
\newtheorem{proposition}{Proposition}
\newreptheorem{proposition}{Proposition}

\usepackage[breaklinks=true]{hyperref}

\hypersetup{
  colorlinks   = true, 
  urlcolor     = blue, 
  linkcolor    = blue, 
  citecolor   = red 
}

\begin{document}

\title{Projected Least-Squares Quantum Process Tomography}

\author{Trystan Surawy-Stepney}
\email[corresponding author: ]{eetss@leeds.ac.uk}
\affiliation{School of Mathematical Sciences, University of Nottingham, United Kingdom}

\author{Jonas Kahn}
\affiliation{Institut de Math\'ematiques de Toulouse and ANITI, Universit\'e de
Toulouse, France}

\author{Richard Kueng}
\affiliation{Institute for Integrated Circuits, Johannes Kepler University Linz, Austria}

\author{Madalin Gu{\c t}{\u a}}
\affiliation{School of Mathematical Sciences, University of Nottingham, United Kingdom}


\maketitle

\begin{abstract}
We propose and investigate a new method of quantum process tomography (QPT) which we call projected least squares (PLS). In short, PLS consists of first computing the least-squares estimator of the Choi matrix of an unknown channel, and subsequently projecting it onto the convex set of Choi matrices.  We consider four experimental setups 
including direct QPT with Pauli eigenvectors as input and Pauli measurements, and ancilla-assisted QPT with mutually unbiased bases (MUB) measurements. 
In each case, we provide a closed form solution for the least-squares estimator of the Choi matrix. We propose a novel, two-step method for projecting these estimators onto the set of matrices representing physical quantum channels, and a fast numerical implementation in the form of the hyperplane intersection projection algorithm. We provide rigorous, non-asymptotic concentration bounds, sampling complexities and confidence regions for the Frobenius and trace-norm error of the estimators. For the Frobenius error, the bounds are linear in the rank of the Choi matrix, and for low ranks, they improve the error rates of the least squares estimator by a factor $d^2$, where $d$ is the system dimension. We illustrate the method with numerical experiments involving channels on systems with up to 7 qubits, and find that PLS has highly competitive accuracy and computational tractability.
\end{abstract}


\section{Introduction}


Quantum process tomography (QPT) -- the task of estimating an unknown quantum transformation completely from measurement data -- is a powerful, but resource-intensive, tool for quantum technology. It has been applied in a variety of experimental contexts \cite{Riebe2006,Weinstein2004,PhysRevLett.93.080502, Pach_n_2015, Bialczak_2010, Howard_2006}. 
However, the ongoing growth in size of quantum hardware, such as quantum circuits and simulators, motivates the design of QPT techniques that use relevant resources as sparingly as possible. 

The earliest QPT proposals \cite{Chuang_1997, PhysRevLett.78.390} used linear inversion techniques to reconstruct quantum channels from informationally complete datasets obtained by feeding different (known) input states into the process and performing full quantum state tomography on the resulting outputs. 
Later, it was realised that by applying the channel-state duality uncovered by Choi \cite{CHOI1975285} and Jamiołkowski \cite{JAMIOLKOWSKI1972275}, one could perform QPT by simply applying the channel to a maximally entangled state on the system and an ancilla, and using quantum state tomography to reconstruct the channel's Choi matrix directly~\cite{leung2000robust,PhysRevLett.86.4195}. 
Other statistical methods developed in the context of state estimation have been adapted for QPT, including maximum-likelihood \cite{PhysRevA.64.052312, Lvovsky_2004, PhysRevLett.105.200504, Smolin_2012} and Bayesian inference \cite{Granade_2016, PhysRevLett.102.020504, Blume_Kohout_2010, Granade_2017} techniques. In practical applications, point estimators of the quantum process need to be accompanied by error bars \cite{Haffner_2005}, hence the importance of developing confidence regions methodology \cite{Christandl&Renner,Faist&Renner,Faist2019}.


A major challenge in current applications is the experimental and computational cost of estimating high dimensional systems. The `low-rank' paradigm developed in the context of compressed sensing \cite{Flammia_2012, Roth2018, Kliesch2019} partly addresses this challenge by reducing the number of measurements required to estimate channels with a small number of Kraus operators, or equivalently, a Choi matrix that is low-rank. Other types of latent structures, such as matrix product states with small bond dimension, have been exploited in quantum tomography \cite{Cramer2009,Baumgratz2013,Lanyon2017} and more recently extended to quantum state and process tomography in conjunction with machine learning \cite{Torlai2018,torlai2020}. Techniques like randomized benchmarking can reliably certify unitary channels with limited resource consumption and a robustness against state preparation and measurement errors \cite{PhysRevLett.102.090502, Knill_2008}. An alternative approach is to estimate the fidelity with respect to a known target process \cite{PhysRevLett.108.260506, PhysRevLett.107.210404}. 

In this work we go back to the traditional `full-tomography' problem and propose a computationally effective and statistically tractable QPT
method, based on an extension of the \emph{projected least-squares} (PLS) method of state estimation proposed in \cite{Gu__2020}. The  method is particularly effective in estimating low rank channels as a proxy for noisy unitary transformations realised in near-term quantum architectures. 
We provide theoretical concentration bounds for the Frobenius and trace-norm errors and and present simulation results for channels with up to 7 qubits (this corresponds to a Choi matrix with $2^{28} \approx 2.684 \times 10^8$ degrees of freedom).

Let us consider an unknown channel $\mathcal{C}:M(\mathbb{C}^d)\to M(\mathbb{C}^d)$, and represent this by its Choi matrix $\Phi = \mathcal{C}\otimes \mathcal{I}_d (\Omega)$, where $\Omega$ is a maximally entangled state on $\mathbb{C}^d\otimes \mathbb{C}^d$. 
We gather data by repeatedly measuring either the output of the channel with known inputs (direct QPT), or the Choi matrix as output of $\mathcal{C}\otimes \mathcal{I}_d$ (ancilla-assisted QPT). The PLS estimator $\hat{\Phi}_{PLS}$ is obtained by first computing the least squares estimator $\hat{\Phi}_{LS}$ of $\Phi$ based on the measurement data, and then projecting this on the convex space of Choi matrices, the latter being the intersection between the cone of positive matrices $\CP$ and the hyperplane $\TP$ determined by the linear constraints ${\rm Tr}_s(\Phi)= \mathbb{1}_d/d$.
We note that the idea of projecting the least-squares estimator of a Choi matrix onto the intersection of $\CP$ and $\TP$ has been previously considered in \cite{PhysRevA.98.062336}. However, this differs in several respects to our implementation: we consider specific choices of input states and measurement settings for which we provide explicit expressions for the LS estimator, and more significantly, we use a different projection algorithm. These result in a protocol of significantly lower computational complexity for which we are able to provide precise statistical convergence results.


We implement the projection using several methods including Dykstra's algorithm \cite{10.2307/2288193} and our own \emph{hyperplane intersection projection} (HIP). Dykstra's algorithm involves an iteration of alternating projections onto $\mathcal{CP}$ and $\mathcal{TP}$ with certain adjustments that make sure that the limit is the closest point with respect to the Frobenius distance. While the projection onto $\mathcal{TP}$ is fast, the projection onto the cone $\mathcal{CP}$ requires matrix diagonalisation and is the slowest sub-routine of the algorithm. The HIP algorithm also alternates between those projections, but uses them to approximate $\CP$ and directly compute the projection on the intersection between this approximation and $\TP$. This vastly decreases the required number of iterations and leads to a significant speed-up in the computation time, while the statistical errors remain comparable to those of PLS using Dykstra's algorithm. To illustrate the versatility of PLS we estimated a 7-qubit noisy version of the quantum Fourier transform, with Pauli measurements (cf. Section \ref{sub:pauli-ancilla} for details of the experimental setup). The channel is obtained by performing the quantum Fourier transform and then measuring the first qubit in the $z$ direction with probability $1/4$. Figure \ref{fig:vignette} contains a summary of the results, showing increasing accuracy in trace-norm error, as well as eigenvalues and eigenvectors estimators, with sample size.  While sample sizes of order $10^{12}$ are prohibitive for experiments, we note that important features become visible for more moderate sample sizes of around $10^8$. For a rank-one channel, $10^7$ would arguably be enough, and is significantly less than the number of measurement settings. 

\begin{figure*}[t]
  \includegraphics[width=\textwidth, center]{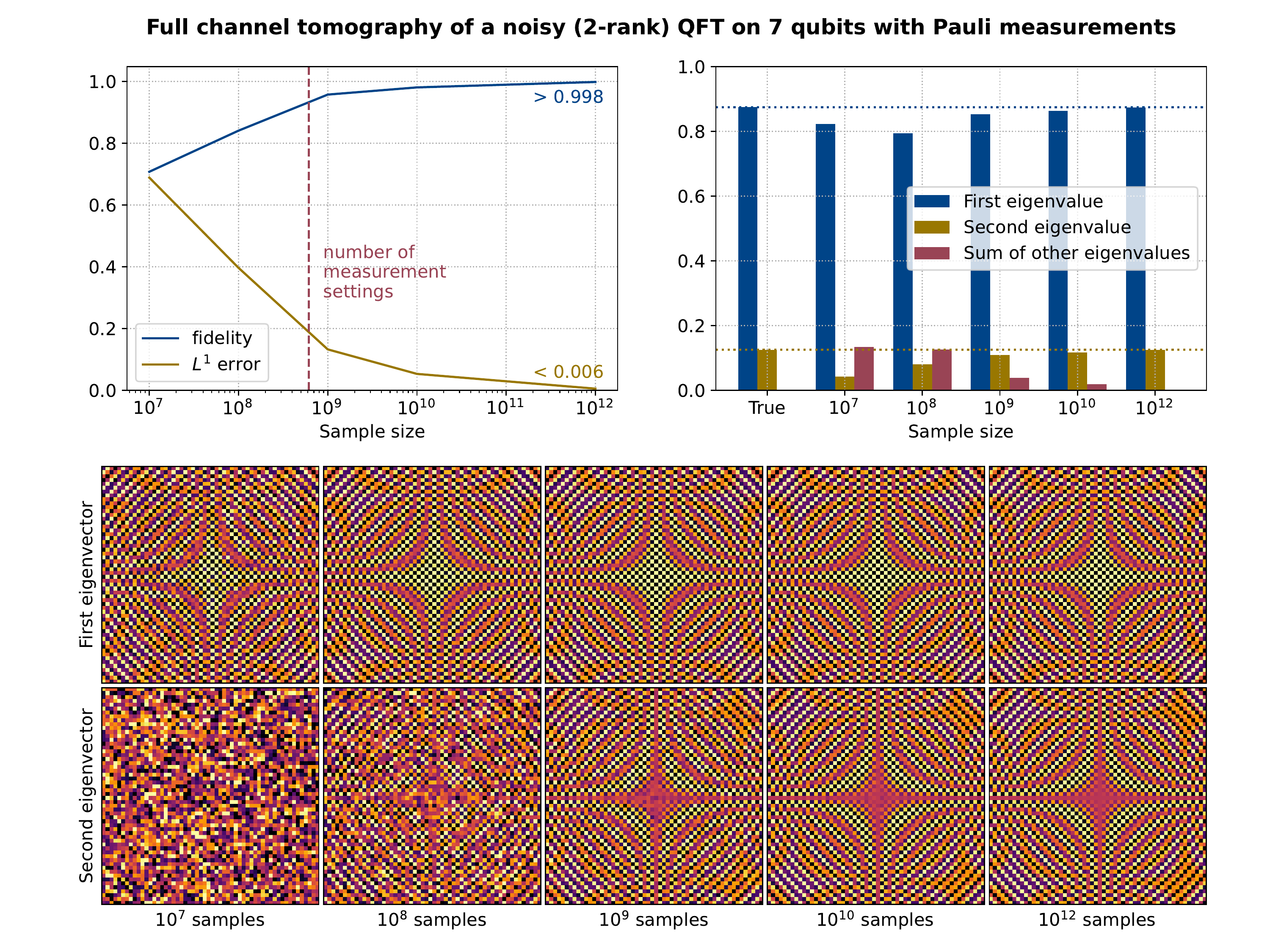}
  \caption{PLS QPT of a 7-qubit noisy quantum Fourier transform (QFT) channel of rank 2 with product input states and Pauli measurements. The top left panel shows the trace-norm error and fidelity as function of sample size, with reasonable errors for one repetition for each combination of input state and measurement setting.
  The top right panel shows the true and estimated eigenvalues of the Choi matrix as function of the sample size.
  The bottom panels show details of
  the real part of the first Kraus operator, and the imaginary part of the second one for different sample sizes.
  The structure of the first eigenvector becomes visible at $10^7$ samples while that of the second at $10^9$.}\label{fig:vignette}
\end{figure*}

The proposed methods are accompanied by mathematically rigorous performance guarantees. 
They are summarised in Theorem~\ref{thm:projected_bounds} and Corollary~\ref{thm:sampling_complextites} which establish concentration bounds for estimators resulting from four possible experimental scenarios detailed in Section \ref{sec:experimental_procedures}. In particular, we show that the number of samples required to reconstruct a rank $r$ Choi matrix to accuracy $\epsilon$ with respect the the Frobenius distance is  $rd^2/\epsilon^2$ (for trace-norm distance it becomes $r^2d^2/\epsilon^2$). This shows that that  PLS is particularly well-suited for estimating approximately unitary  channels ($r \approx 1$), for which it achieves error reduction by a factor $d^2$ compared to the ordinary LS estimator. In Theorem \ref{thm:projected_bounds_general} we provide a recipe for constructing confidence regions without any prior assumption about the Kraus rank $r$.\par

\vspace{2mm}
\paragraph*{\textbf{Roadmap}}
Section \ref{sec:background}  summarises important background, in particular the LS estimator.
Section \ref{sec:experimental_procedures} describes the 4 experimental procedures we analyse. 
Each procedure concerns either direct or ancilla-assisted QPT and involves either Pauli measurements or mutually unbiased bases measurements. Explicit expressions for each LS estimator are also provided. Section \ref{sec:methods_of_projection} introduces the PLS estimator and its different implementations. Section \ref{sec:error_bounds_and_scs} contains the main theoretical results described above. Section \ref{sec:numerical_expts} details the numerical algorithms and present numerical results comparing different projection methods, the dependence of errors on rank, and that of the computation time on dimension. Detailed proofs of concentration bounds for LS and PLS estimators can be found in the appendices.

\section{Background}\label{sec:background}
In quantum mechanics,
the state of a $d$-dimensional quantum system is represented by a density matrix, i.e. a positive operator $\rho\in M(\mathbb{C}^d)$ of trace one ($\Tr(\rho)=1$).
This state characterizes all properties and correlations of the underlying quantum system, but is not directly accessible.
One can gain information about it by measuring the system. A measurement with a finite number of outcomes, say $\{1,2,\dots, s\}$, is described by a positive operator valued measure (POVM). This is a set of positive operators $\{M_i\}_{i=1}^s$ on $\mathbb{C}^d$, such that $\sum_i M_i= \mathbb{1}_d$, where  $\mathbb{1}_d$ is the identity operator. The measurement outcome is random and 
the underlying probability distribution $p$ is given by Born's rule:
$$
p= \{p_i\}_{i=1}^s,  \qquad{\rm with~} 
p_i = {\rm Tr} (\rho M_i).
$$  
In quantum state tomography (QST) one aims to estimate an unknown state $\rho$ by  measuring a (large) ensemble of independent state copies prepared in the state $\rho$. Among various measurement and estimation techniques described in the literature \cite{Paris}, the projected least squares (PLS) method \cite{sugiyama_least-squares_2013,Gu__2020} is particularly suitable for estimating states which are well approximated by low rank operators.
In a nutshell, PLS combines an initial `quick and dirty' least squares estimation step with a `noise reduction' step where the unphysical estimator is projected onto the space of quantum states. For low-rank states, the projection step reduces the estimation error by a factor of order $d$, and PLS was shown to achieve optimal error rates when the measurements are chosen from a sufficiently generic ensemble (2-design).
 
\vspace{2mm}

In this paper we address the related problem of quantum process tomography (QPT). Rather than estimating a quantum state, we want to estimate 
a physical evolution, such as that implemented by a noisy quantum circuit.    Such a transformation is described by a quantum channel, i.e. a completely positive, trace preserving map
$$
\mathcal{C}:M(\mathbb{C}^{d})\rightarrow M(\mathbb{C}^{d}).
$$
In QPT we assume that the channel is unknown, and we would like to estimate it by probing many identical realisations of $\mathcal{C}$ with known input states and performing measurements on the outputs. 

A particularly useful representation of $\mathcal{C}$ is the \emph{Choi matrix} \cite{CHOI1975285}. 
It arises from maximally entangling the input intended for $\mathcal{C}$ with an equally-sized quantum memory, also called an ancilla.
$$
\Omega :=\ket{\omega}\!\bra{\omega},\; \text{where} |\omega\rangle:=
\frac{1}{\sqrt{d}} \sum_{q =1}^d\ket{q} \otimes\ket{q} \in \mathbb{C}^{d}\otimes \mathbb{C}^{d}
$$
is a maximally entangled state on two $d$-dimensional quantum systems.
Here, $\{ |q\rangle \}$ denote the 
computational basis vectors in $\mathbb{C}^{d}$. 
By applying the channel $\mathcal{C}$ to the system
and leaving the quantum memory untouched, we obtain the Choi matrix  of the channel
$$
\Phi = \Phi_\mathcal{C} = (\mathcal{C} \otimes \mathcal{I}_d) (\Omega). 
$$
Here, $\mathcal{I}_d$ denotes the identity operation on the quantum memory (do nothing).
 The Choi matrix describes a bipartite quantum state which satisfies the constraint ${\rm Tr}_s ( \Phi) = \mathbb{1}_d/d$, where ${\rm Tr}_{s/a}$ denotes the partial trace over the 
 system or ancilla respectively. 
 Conversely, any positive operator satisfying the above constraint is the Choi matrix of a quantum channel. The action of the quantum channel $\mathcal{C}$ on a state $\rho$ can be expressed in terms of its Choi matrix as
 \begin{equation}
 \label{eq.choi.channel.action}
\C(\rho) = {\rm Tr}_{a}(\Phi\: (\mathbb{1}_d\otimes\rho^{\top}) ),  
 \end{equation}
where $\top$ denotes the transpose. 

Quantum channel and Choi matrix are two mathematically equivalent descriptions of the underlying process.
Moving from one to the other allows us to recast QPT as an instance of QST. 
Meaningful distance measures between channels, say $\C_1$ and $\C_2$, can also be expressed in terms of differences between the Choi matrices $\Phi_1$ and $\Phi_2$. This includes the squared Frobenius distance
\begin{equation*}
\|\Phi_1 -\Phi_2 \|_2^2= {\rm Tr} \left[ (\Phi_1 -\Phi_2 )^2\right],
\end{equation*}
as well as trace distance 
\begin{equation*}
\| \Phi_1-\Phi_2 \|_1 = {\rm Tr} (| \Phi_1-\Phi_2| ),  \quad |X| = \sqrt{X^*X}.
\end{equation*}
The latter is related to the diamond distance between the underlying channels -- arguably one of the strongest and most widely used distance measures for quantum channels, see e.g.\ \cite[Proposition~50]{kliesch_tutorial_2021}. 

We restrict our attention to independent and sequential uses of the channel. 
This is an actual restriction, as there are powerful techniques for QST \cite{Kahn_2009,Wright2016,Haah2017}, as well as 
quantum enhanced metrology 
\cite{Giovannetti_2011,Zhou_2021},
which require the parallel application of the channel to different parts of a global quantum system. 
We do, however, investigate both ancilla-free and ancilla-assisted strategies. 

In the first case the system is measured after an appropriate input state $\rho$ has been passed through the channel $\mathcal{C}$. In this case the probabilities associated to an output measurement with POVM $\{M_i\}_{i=1}^s$ are 
$$p_i = {\rm Tr}(\mathcal{C}(\rho)M_i)=
{\rm Tr}\left(\Phi (M_i \otimes \rho^\top) \right),
$$
where the second equality uses Eq.~\eqref{eq.choi.channel.action}. This procedure is illustrated in Figure~\ref{fig:noAAPT}, Section \ref{sub:pauli-direct}.

In the ancilla-assisted case, we input one half of a maximally entangled state $\Omega$ into the unknown channel, the other half is left unchanged in a quantum memory. This procedure is illustrated in Figure~\ref{fig:AAPT}, Section \ref{sub:pauli-ancilla} and produces an output quantum state described by the Choi matrix $\Phi$. Subsequently, we perform quantum measurements. The outcome probabilities associated with POVM $\{M_i\}_{i=1}^s$ are
 $$
 p_i= {\rm Tr}(\Phi M_i).
 $$

A common feature of both setups described above is the existence of a positive linear transformation 
$$
\mathcal{A}: M(\mathbb{C}^d) \otimes M(\mathbb{C}^d) \to \mathbb{C}^s
$$
which maps the Choi matrix $\Phi$ into the probability vector of the measurement outcomes. That is, $p= \mathcal{A} (\Phi) \in \mathbb{C}^s$, where $s$ counts the number of measurement outcomes. This follows from the fact that the measured states depend linearly on the channel which, in turn, is in a linear, one-to-one correspondence with its Choi matrix. 

After performing several independent preparation-measurement rounds, the data is collected as a \emph{vector of frequencies} $f\in\mathbb{R}^s$, so that for for large sample sizes $f\approx p$. The resulting estimation problem can be recast as linear regression
$$
 f = p+ \epsilon = \mathcal{A} (\Phi) + \epsilon,
$$
where $\epsilon$ is the statistical noise due to finite sample size. The simplest estimator for $\Phi$ is the least squares (LS) estimator:
\begin{equation}\label{eq:ls_estimator}
\hat{\Phi}_{LS} =\arg\min_{\tau}  \| \mathcal{A} (\tau) - f\|  = (\mathcal{A}^\dagger \mathcal{A})^{-1} \mathcal{A}^{\dagger}(f),
\end{equation}
where $\mathcal{A}^\dagger$ is the adjoint of $\mathcal{A}$. The minimum is taken over all matrices $\tau$ on $\mathbb{C}^d \otimes \mathbb{C}^d$ with compatible dimension,  and the loss function  $\|\cdot\|$ is the $\ell_2$-norm.\\\par

The solution to Problem~\eqref{eq:ls_estimator} is easy to compute and produces a self-adjoint matrix with unit trace.
However, $\hat{\Phi}_{LS}$ is typically not a Choi matrix of a channel. A quantum channel must be \textit{trace-preserving} and \textit{completely positive}; properties which provide additional constraints on the estimators of the Choi matrix. To obtain estimators representing physical quantum channels, we propose to project the 
LS estimators onto the set of matrices that satisfy the given constraints. 


By Choi's theorem \cite{CHOI1975285}, the complete positivity condition for the quantum channel is equivalent to the  \textit{positive semidefiniteness} of the $d^2 \times d^2$ Choi matrix. The convex set of positive-semidefinite matrices forms a pointed cone in the space of complex square matrices. We call this cone $\mathcal{CP}$, because it contains matrices that represent completely positive maps. The trace-preserving property of the quantum channel is equivalent to a partial trace condition for the Choi matrix: $\mbox{Tr}_{s}\Phi=\mathbb{1}_{d}/d$.
This provides a further affine constraint involving $d^{2}$ elements. Hence, the trace-preserving matrices lie on a flat plane of dimension $d^{4}-d^{2}$. 
We call this set $\mathcal{TP}$.\par
The intersection of $\C\mathcal{P}$ and $\T\mathcal{P}$ forms the set of matrices 
that represent \emph{physical quantum channels}.
Accordingly, we call this intersection $\C\mathcal{P}\mathcal{T}\mathcal{P}$ and refer to Figure \ref{fig:pls_geom_white} for a geometric illustration.
\begin{figure}[!htb]
 \begin{center}
   \includegraphics[width=0.5\textwidth]{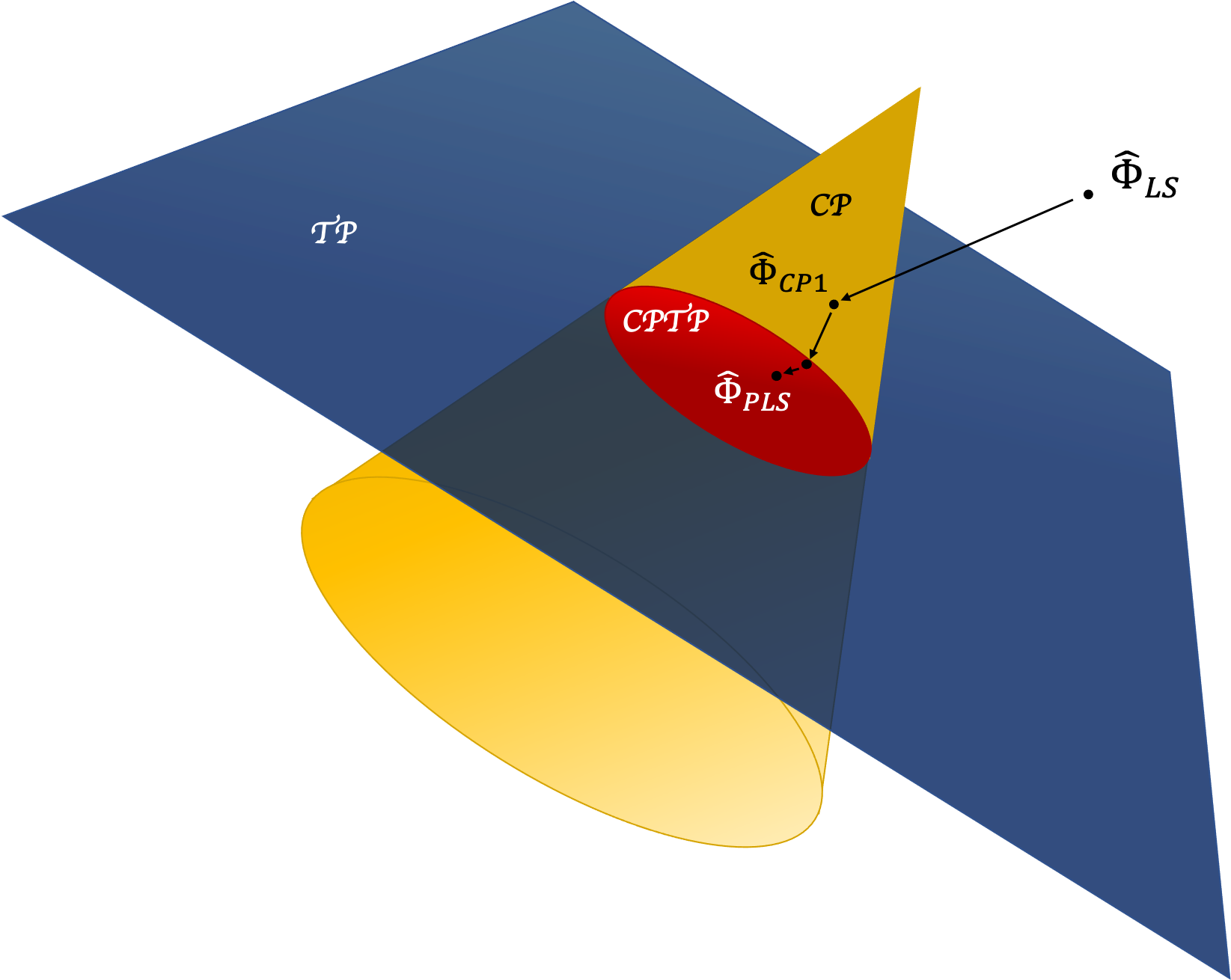}
 \end{center}
 \caption{The set of CP maps on $M(\mathbb{C}^{d})$ are represented via the Choi isomorphism by a pointed cone embedded in $M(\mathbb{C}^{d^2})$. The set of trace preserving maps forms a 
 hyperplane with codimension $d^{2}$, whose intersection with the CP cone is the set of Choi matrices. We also show a representation of the projection of the LS estimator onto the set of quantum states, giving $\CPone$, followed by the projection onto $\C\mathcal{P}\mathcal{T}\mathcal{P}$, including mixing with $\mathbb{1}_d/d$, with solution $\hat{\Phi}_{PLS}$.}\label{fig:pls_geom_white}
\end{figure}

For concreteness, we focus on reconstructing physical evolutions of multi-qubit systems. That is, the unknown channel $\mathcal{C}$ acts on a quantum system comprised of $k$ qubits, i.e.\ $d=2^k$. This will allow us to discuss different scenarios depending on the choice of measurement and input states. We will focus on two types of measurement setups: single-qubit Pauli measurements and (global) mutually unbiased basis measurements. The different choices of experimental strategies (ancila-free or ancilla-assisted) and measurement setups (single-qubit Pauli or mutually unbiased bases) produce four scenarios which are discussed in detail in the following section. For each scenario, we give a closed form solution to Eq. \eqref{eq:ls_estimator}.

\section{Proposed experimental procedures and least-squares estimators}\label{sec:experimental_procedures}

\subsection{Scenario 1: ancilla-assisted QPT with single-qubit Pauli measurements} \label{sub:pauli-ancilla} 
We first propose an ancilla-assisted protocol. The experimenter prepares many independent copies of the maximally entangled state $\Omega$ and acts with the quantum channel on the right subsystem consisting of $k$ qubits, before performing state tomography on the collection of output states which are instances of the Choi matrix $\Phi$ (Figure \ref{fig:AAPT}). 

\begin{figure}[!htb]
 \begin{center}
   \includegraphics[width=0.45\textwidth]{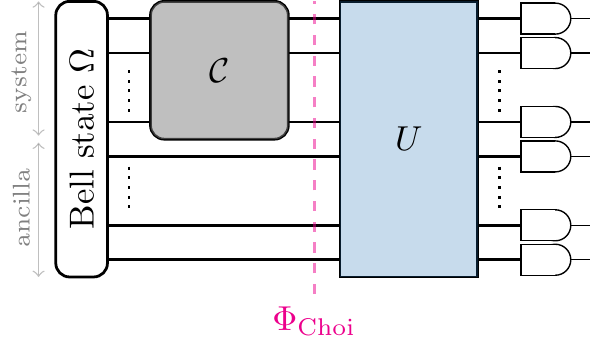}
 \end{center}
 \caption{Ancilla-assisted QPT.}\label{fig:AAPT}
\end{figure}

The experimenter performs single-qubit Pauli measurements on each of the $2k$ qubits. That is, each qubit is measured either in the (Pauli) $x$-, $y$- or $z$-basis.
Let 
$\{|0,s\rangle, |1,s\rangle\}$ denote the eigenbasis of the Pauli matrix $\sigma_s$, i.e.\ $\sigma_s |o,s\rangle =(-1)^o |o,s\rangle$ for $s =x,y,z$ (Pauli matrix) and $o =-1,+1$ (eigenvalue).
In particular, $|+1,z\rangle = |0 \rangle$ and $|-1,z \rangle = |1\rangle$ (computational basis), while $|+1,x \rangle = |+ \rangle$ and $|-1,x \rangle = |- \rangle$ (Hadamard basis).
To measure the Pauli observable $\sigma_s$ we apply a single-qubit basis change $U_s = \sum_{o=0}^1|o,s\rangle \langle o,s|$ and, subsequently, measure the qubit in the standard basis. For a given collective setting ${\bf s} = (s_1,\dots , s_{2k})\in \{x,y,z\}^{2k}$, the corresponding $2k$-qubit unitary factorizes nicely into tensor products: $U= \otimes_{i} U_{s_i}$, see Figure \ref{fig:AAPT}. The outcome of such a measurement is a sequence ${\bf o}= (o_1,\dots, o_{2k}) \in \{0,1\}^{2k}$, and the corresponding POVM is the set of one dimensional projections 
$P^{\bf s}_{\bf o} = \otimes_{i=1}^{2k}|o_i,s_i\rangle\langle o_i,s_i|$.

Two schemes for generating measurement settings come to mind: a `fixed' one and a `random' one. In the fixed scheme, we cycle
through all $3^{2k}$ possible measurement settings a total of $\nu$ times. There, sample size $N$ transforms to $\nu 3^{2k}$ cycles. In the random scheme, we choose measurement settings uniformly at random, and write $\nu = N/3^{2k}$ for the mean number of times each measurement setting is used. In the random scheme $\nu$ may be non-integer and strictly smaller than one. 

In both cases, the frequency $f^\vs_\vo$ is equal to the number of times the outcome \vo\: has occurred when measuring in setting \vs, divided by $\nu$. 

These frequencies are unbiased estimators of, and tend to in the limit of large $N$, the probabilities of observing outcome \vo\:when measuring with setting \vs, given by Born's rule:
\begin{equation}\label{sc1_probs}
p^\vs_\vo=\mbox{Tr}(\Phi P^\vs_\vo).
\end{equation}

The formula of the least-squares estimator can be adapted from \cite{Gu__2020} which deals with QST from single-qubit Pauli measurements:
\begin{equation}\label{ls_estimator_sc_1}
\hat{\Phi}_{LS}=\frac{1}{3^{2k}}\sum_{\textbf{s}, \textbf{o}}f^{\textbf{s}}_{\textbf{o}}\:\overset{2k}{\underset{i=1}{\bigotimes}}(3\ket{o_i, s_i}\!\bra{o_i, s_i}-\mathbb{1}_2 ).
\end{equation}
Here, the $|o_i,s_i \rangle$ are single-qubit eigenstates of a Pauli matrix, and $\mathbb{1}_2$ is the 2-dimensional identity operator.
\par

\subsection{Scenario 2: Direct QPT with single-qubit Pauli measurements}\label{sub:pauli-direct}
Compared to preparing a large number of maximally entangled states of $2k$ qubits, 
individual $k$-qubit (pure) product states are a much more tractable proposition. 
Our second proposed experimental setup, illustrated in Figure~\ref{fig:noAAPT}, takes this into account. 

\begin{figure}[!htb]
 \begin{center}
   \includegraphics[width=0.65\textwidth]{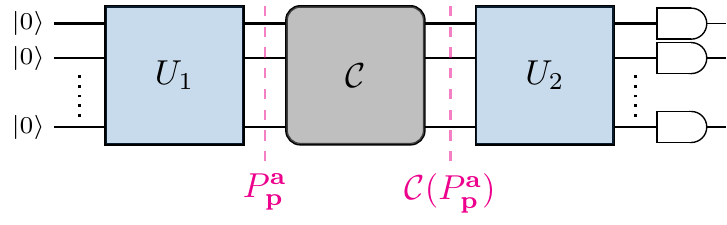}
 \end{center}
 \caption{Direct QPT (without ancilla)
 }\label{fig:noAAPT}
\end{figure}\par

The experimenter prepares a state from the set of pure product states $\{{P^{\va}_{\vp}}^{\top}\}$ where $\va\in\{x,y,z\}^{k}$ and $\vp\in\{0,1\}^{k}$, and $\top$ denotes the transpose. This is achieved by initialising the qubits in the product state 
$|{\bf 0}\rangle= \otimes_{i=1}^k |0\rangle$ and applying a 
unitary $U^{a_i}_{p_i}$ to qubit $i$ such that
$$
U^x_p|0\rangle\mapsto |p, x\rangle, \:
U^y_p|0\rangle\mapsto |p\oplus 1 , y\rangle, \:
U^z_p|0\rangle\mapsto |p, z\rangle
$$
The overall product unitary is $U =\otimes_{i=1}^k U^{a_i}_{b_i}$.

The quantum channel $\mathcal{C}$ is applied to the prepared state producing the output state $\C({P^\va_\vp}^\top)$. 
Using another product unitary $U_2$ as described in the previous section, the experimenter performs a Pauli measurement with measurement setting ${\bf b}\in\{x,y,z\}^{k}$, to obtain an outcome ${\bf q}\in\{0,1\}^{k}$. Similarly to scenario 1, either we repeat the procedure a fixed number of times $\nu$ for each choice of indices ${\bf a},{\bf b},{\bf p}$, so that the sample size is $N = \nu 3^{2k}2^k$, or we choose the indices at random for each of the $N$ samples and set $\nu =  N/(3^{2k} 2^{k})$ for the mean number of times each setting is used.
The \vq-th frequency entry $\{f^{\va\vb}_{\vp\vq}\}_{\vq}$ counts the number of times outcome \vq\:is observed having measured the state $\C({P^\va_\vp}^\top)$ with setting \vb, divided by $\nu$.
These frequencies are unbiased estimators of the probabilities
\begin{equation}
p^{\va\vb}_{\vp\vq}=\mbox{Tr}(\C({P^\va_\vp}^\top)P^{\vb}_{\vq}) 
\end{equation}
and converge to them in the limit of large $N$.

\begin{proposition}\label{thm:ls_sc2}
For direct QPT with single-qubit Pauli inputs and single-qubit Pauli measurements, the least-squares estimator for the Choi matrix $\Phi$ of a  $k$-qubit quantum channel takes the following form:
\begin{align}\label{ls_scenario_2}
\hat{\Phi}_{LS}=&\frac{1}{3^{2k}d}\sum_{\va\vb\vp\vq}f^{\va\vb}_{\vq\vp}\overset{k}{\underset{i=1}{\bigotimes}} M\:^{b_i}_{p_i} \,\overset{k}{\underset{j=1}{\bigotimes}}M\:^{a_j}_{q_j}, \quad \text{where} \\
& M\:^{b_i}_{p_i}= (3\ket{p_i,b_i}\!\bra{p_i,b_i}-\mathbb{1}_2), \\
& M\:^{a_j}_{q_j}=(3\ket{q_j,a_j}\!\bra{q_j,a_j}-\mathbb{1}_2).
\end{align}
\end{proposition}

We refer to Appendix~\ref{subsec:proof_of_ls2} for a proof.

\subsection{Scenario 3: Ancilla-assisted QPT with mutually unbiased basis measurements} \label{sub:MUB+ancilla}
The set-up is similar to that of ancilla-assisted QPT with single-qubit Pauli measurements, though we now use the POVM consisting of a maximal set of mutually unbiased bases (MUB) \cite{Schwinger1960,Wootters1989} to measure instances of the Choi matrix. Such measurements require entangling gates to implement, but they also have a rich, global structure that turns out to produce sampling complexities that are essentially optimal \cite{Gu__2020}.
A MUB POVM on $2k$ qubits consists of $m = d^2(d^2+1)$ 
rank-one operators $\{M_k =|v_k\rangle \!\langle v_k|/(d^2+1)\}_{k=1}^m$ in $d^2=2^{2k}$ dimension. Each vector (range) $|v_k\rangle$ belongs to one of $d^2+1$ orthonormal bases,  and any pair belonging to different bases satisfies the unbiasedness condition $|\langle v_i |v_j\rangle|^2 = d^{-2}$. The following `near isotropy' property is shared with a larger class of 2-design POVMs, see e.g.\ \cite[Lemma~8]{Gross2015},
and plays an important role in proving our results: 
$$
 \sum_{k=1}^m |v_k\rangle \!\langle v_k|  {\rm Tr} (|v_k\rangle\! \langle v_k| A) = 
 A + {\rm Tr} (A) \mathbb{1}_d^{\otimes 2},
$$
for all operators $A$ with compatible dimensions.
Our proofs rely mainly on the near isotropy property
and can therefore be readily generalized to other (approximate) 2-design constructions. Concrete examples are symmetric informationally complete (SIC) POVMs \cite{Renes_2004}, complete sets of stabiliser states \cite{PhysRevA.80.012304} and local random circuits of depth (order) $k$ \cite{brandao2016local,hunterjones2019unitary,PhysRevA.104.022417}.

Let us now return to the actual protocol. The system-ancilla input state $\Omega$ is passed through the channel $\mathcal{C}\otimes\mathcal{I}_d$ and produces
a system characterized by the Choi matrix $\Phi$ as output. This system is then measured with a MUB POVM as described above. The frequency $f_i $ indicates the number of times outcome $i$ is observed, divided by the total number of measurements $N$.
These frequencies are unbiased estimators of, and tend to in the limit of large $N$, the probabilities
\begin{equation}
p_i=\mbox{Tr}(\Phi M_i).
\end{equation}
As with single-qubit Pauli measurements, the expression of the least squares estimator is found in \cite{Gu__2020}
\begin{equation}\label{ls_estimator_sc_real_3}
\hat{\Phi}_{LS}=(d^2+1)\sum_{i=1}^{m} f_i\ket{v_i}\!\bra{v_i}\:-\: \mathbb{1}_{d^2}.
\end{equation}\par

\subsection{Scenario 4: Direct QPT with mutually unbiased basis measurements} \label{sub:MUB-direct}

Our final experimental method is similar to direct QPT with single-qubit Pauli measurements, in that we do not use an ancillary system. However, we use a different POVM and a different set of input states.\par
As input states we use the transposes of one dimensional projections $|w_k\rangle\!\langle w_k|$ onto vectors of a set of $d+1$ MUBs. On the respective outputs, we measure a MUB POVM whose elements $|v_l\rangle\!\langle v_l|$ may be different from those of the input. 
As with the previous scenario, similar results can be obtained with other 2-designs.

Here, frequencies $f^{k}_{l}$ indicate the number of times outcome $l \in\{ 1,\dots ,m\}$ is observed having measured the state $\C(\ket{w_k}\!\bra{w_k}^\top)$, divided by the (mean) number of repetitions $\nu=\frac{N}{m}$. 
These frequencies are unbiased estimators of the probabilities
\begin{equation}
p^{k}_{l}=\mbox{Tr}(\C(\ket{w_k}\!\bra{w_k}^\top)\:M_l), 
\end{equation}
with $M_l = \frac{1}{d+1}|v_l \rangle \! \langle v_l|$
and converge to them in the limit of large $N$.\par

\begin{proposition} \label{th.LS.scenario4.general}
For direct QPT with (transposed) MUB inputs $(|w_i \rangle \! \langle w_i|)^\top$ and MUB measurements $\tfrac{d}{m}|v_l \rangle \! \langle v_l|$, the least-squares estimator for the Choi matrix $\Phi$ of a $k$-qubit quantum channel takes the following form: 
\begin{eqnarray}
\hat{\Phi}_{LS}  &=&
\frac{d+1}{d} \sum_{l,k=1}^m f_l^k   \ket{v_l}\!\bra{v_l} \otimes   \ket{w_k}\!\bra{w_k} \nonumber \\
&& - \frac{1}{d}\sum_{l,k=1}^m f_l^k  (  \ket{v_l}\!\bra{v_l} \otimes \mathbb{1}_d + \mathbb{1}_d \otimes  \ket{w_k}\!\bra{w_k} )
\nonumber \\
&&+ \mathbb{1}_d  \otimes  \mathbb{1}_d.
\label{eq.LS.scenario4.general}
\end{eqnarray}
\end{proposition}

We refer to Appendix~\ref{proof.th.LS.scenario4.general} for a proof.\\\par

%
%
%

%
%

We conclude this section by emphasizing that all estimators found from \eqref{eq:ls_estimator} are linear and unbiased, given that the frequencies observed are unbiased estimators of true probabilities. Hence those estimators discussed so far, $\hat\Phi_{LS}$, are unbiased estimators of the underlying Choi matrix $\Phi$.



\section{Methods of projection}\label{sec:methods_of_projection}

We turn our attention now to the task of finding our final estimator, $\PLS \in \CPTP$, from the least-squares estimator $\hat\Phi_{LS}$. We do this in the following way (see Figure \ref{fig:pls_geom_white} for an illustration of the geometry):

\begin{enumerate}
    \item Project $\LSE$ onto the set of $\CP$ matrices of trace one (the quantum states) - we call this set $\CP1$ and the resulting estimator $\CPone$.
    \item `Project' $\CPone$ onto $\CPTP$ to obtain $\PLS$.
\end{enumerate}
For reference, the projection of a matrix $M$ onto a non-empty closed convex set $\mathcal{S}$ with respect to a norm $\|\cdot\|_{\alpha}$ is defined by:
$
M_{proj} = \argmin_{M^{\prime}\in\mathcal{S}}\|M-M^{\prime}\|_{\alpha}.
$\par
In order for our theoretical results to hold, the above two procedures need not constitute \textit{true} projections. For our purposes an adequate `projection' would be to find points in $\CP1$ and $\CPTP$ respectively that satisfy the following two key properties:
\begin{align}
    \label{prop_CP1}
    \| \Phi - \hat{\Phi}_{CP1} \|_{\infty} & \leq 2  \| \Phi - \hat{\Phi}_{LS} \|_{\infty} ,
    \\
    \label{decrease_dist}
       \| \Phi - \hat{\Phi}_{PLS} \|_{2} & \leq  \| \Phi - \hat{\Phi}_{CP1} \|_{2} ,
\end{align}
where $\Phi$ is the Choi matrix we are trying to estimate and here, and henceforth, $\|\cdot\|_\infty$ indicates the operator norm.
While a direct, Frobenius projection of $\LSE$ onto $\CPTP$ (e.g. using Dykstra's algorithm) may seem less involved than the above procedure, our method turns out to be more amenable to deriving \emph{a priori} error bounds that are capable of resolving latent structure in the form of low rank. In addition, the advantage of working with more relaxed definitions is that we obtain `projections' that are easier to compute numerically, and  better-behaved in practice. 
For the rest of the paper we focus on the two-step method, but some results also hold for the direct projection method (see Appendix~\ref{app:direct_projection_results}). We  give the general description of steps 1 and 2 and point to the numerics Section \ref{sec:numerical_expts} for the detailed algorithms.\\\par

\emph{Step 1.} The projection onto $\mathcal{C}\mathcal{P}1$ is implemented by an eigenvalue thresholding algorithm (cf. Section \ref{sub:CPthres} for the details and the proof of property \eqref{prop_CP1}). This step constitutes a true projection.

\emph{Step 2.}
Finding $\PLS$ from $\CPone$ is the problem of projecting a matrix onto the intersection of two convex sets ($\CP$ and $\TP$) and is the most involved part of the PLS algorithm. We sketch 3 possible, iterative procedures, including our preferred HIP method, and refer to Section \ref{sub:HIP} for detailed algorithms and the proof of property~\eqref{decrease_dist}. The key idea behind the contraction property (\eqref{decrease_dist}) is contained in the following lemma whose proof can be found a number of textbooks, e.g. \cite{bauschke2011convex}, Prop 4.16.
\begin{lemma}\label{convex_projection_lemma}
Let $S$ be a closed, non-empty, convex subset of $M(\mathbb{C}^{d^2})$ and let $\mathcal{P}_S: M(\mathbb{C}^{d^2})\to M(\mathbb{C}^{d^2})$ be the projection onto $S$ with respect to the Frobenius distance. Then the projection of a matrix $A$ onto $S$ is closer to every point $B$ in $S$ than $A$, i.e.
$$
\| \mathcal{P}_S (A) -B\|_2\leq \|A-B\|_2.
$$
\end{lemma}
Using Lemma \ref{convex_projection_lemma}, one may construct a `projection' of a matrix $X$ onto $\CPTP$, satisfying property \eqref{decrease_dist}, by repeatedly projecting $X$ onto a set of convex sets containing $\CPTP$. 
One such choice is to project $X$ alternately onto $\TP$ and $\CP$. We refer to this as alternating projections (AP). AP can be made into a true projection onto $\CPTP$ by adding a correction term orthogonal to $\CP$ after each $\CP$ projection -- this is Dykstra's algorithm \cite{10.2307/2288193}. In this work, we define our own fast, generalised algorithm which makes use of Lemma \ref{convex_projection_lemma} -- called \emph{hyperplane intersection projection} (HIP), a discussion of which is deferred to Section \ref{sub:HIP}. Again, this does not constitute a true projection but satisfies property \eqref{decrease_dist}. We find this algorithm to be significantly faster than AP or Dykstra's algorithm, with continued convergence at large iterations (see Section \ref{sec:algo_comp}.).\\\par

Whichever of the above iterative methods we chose, after a number of iterations we need to ensure that $\PLS$ exactly satisfies the conditions to be an element of $\CPTP$. For this, we mix the current estimator with the maximally mixed state. Concretely, after near-convergence of the projection algorithm, we have an estimator $\hat{\Phi}^\prime_{PLS}$ in $\TP$ but which may have very small negative eigenvalues. We call the largest magnitude of these $\lambda_{\min}$. For a system of dimension $d$, we thus then take our final estimator to be
\begin{equation}
    \hat{\Phi}_{PLS} = (1-p)\hat{\Phi}^\prime_{PLS} + \frac{p}{d^2}\mathbb{1}_d
\end{equation}
where $p$ is the solution to $(1-p)\lambda_{\min}+p/d^2=0$. In this way, we ensure the estimator is positive-semidefinite while simultaneously preserving the partial trace condition. Mixing with the maximally mixed state adds an error proportional to $\|\lambda_{\min}\|$, hence this additional error can be made arbitrarily small by stopping the algorithm when $\lambda_{\min}$ is sufficiently small in magnitude. The convergence rate of the algorithm is hence defined by how quickly $\|\lambda_{\min}\|$ falls below the required threshold.\\\par

Closed form expressions for projections onto $\CP$, $\TP$ and $\CP1$ are required for the implementation of any protocol discussed in this work. These are provided and proven in Appendix \ref{app:details_on_num_expts}.

\section{Error bounds and sampling complexities}\label{sec:error_bounds_and_scs}

In this section we provide rigorous error bounds on the PLS estimator in the four experimental scenarios described in Section \ref{sec:experimental_procedures}. We first state our results in terms of concetration bounds for the squared Frobenius norm distance  $\|\hat{\Phi} - \Phi\|_2^2$ and the trace norm ($L^1$) distance $\|\hat{\Phi} - \Phi\|_1$, and then discuss how such results can be used to construct confidence regions. A discussion of different channel distance measures, and how they compare to each other, can be found in \cite{Gilchrist_2005}, see also \cite{watrous_2018,Wallman2015, Kueng2016}.\par

\begin{theorem}\label{thm:projected_bounds}

Let  $\mathcal{C}$ be a $k$-qubit channel  and assume that its Choi matrix $\Phi$ has rank $r$. The squared Frobenius norm error of the estimators $\hat{\Phi}_{PLS}$ derived from the LS estimators $\hat{\Phi}_{LS}$ in scenarios \ref{sub:pauli-ancilla}, \ref{sub:pauli-direct}, \ref{sub:MUB+ancilla}, and \ref{sub:MUB-direct} satisfy the following error bounds: for $\epsilon \in (0,1)$,
 \textup{\begin{equation}\label{condensed_error_bound}
 \mbox{Pr}[\|\hat{\Phi}_{PLS}-\Phi\|_{2}\geq\epsilon]\:\leq\:2^{2k}\mbox{exp} \left(-\frac{3N\epsilon^2}{8} \frac{g(k)}{8 r} \right),
 \end{equation}}
 and, respectively,
 \textup{\begin{equation}\label{concentration_l1}
 \mbox{Pr}[\|\hat{\Phi}_{PLS}-\Phi\|_1\geq\epsilon] \leq 
2^{2k}\exp\left( -\frac{3N \epsilon^2}{32} \frac{g(k)}{8r^2}\right).
 \end{equation}}
 Here, $g(k)$ is a function that depends on the number of qubits $k$ and the type of measurement scenario:
 \textup{\begin{equation}
 \label{eq.def.g(k)}
 \begin{aligned}
 g(k)=
 \begin{cases}
 \frac{1}{3^{2k}}&\mbox{Scenarios \ref{sub:pauli-ancilla}, \ref{sub:pauli-direct}},\\
 \frac{1}{2} \frac{1}{2^{2k}}&\mbox{Scenario \ref{sub:MUB+ancilla}},\\
 \frac{1}{4}\frac{1}{2^{2k}}&\mbox{Scenario \ref{sub:MUB-direct}}.
 \end{cases}
 \end{aligned}
 \end{equation}}
 \end{theorem}\par
 \hspace{1cm}

 Theorem \ref{thm:projected_bounds} shows that the PLS estimators can be brought arbitrarily close to the true Choi matrix,  
 provided that we perform sufficiently many measurements. In particular, the Frobenius square error rates scale as 
 $O(r/(Ng(k)))$ up to logarithmic factors, and we note that scenarios \ref{sub:pauli-ancilla} and \ref{sub:pauli-direct} exhibit error bounds which are larger by a factor $(3/2)^{2k}$ than those of scenarios \ref{sub:MUB+ancilla} and  \ref{sub:MUB-direct}. This is consistent with the fact that MUB measurements 
 are more `informative' than local Pauli measurements that necessarily factorize in to tensor products \cite{Kueng2017,Huang2020,Franca2020}. Another interesting observation is that the error rates for ancilla-assisted versus ancilla-free settings are the same in the case of Pauli measurements and differ only by a factor $2$ in the case of MUB measurements.   

The next corollary provides the sampling complexities derived from the above the theorem.

\begin{corollary}\label{thm:sampling_complextites}
 Fix $\epsilon, \eta \in (0,1)$.
 To achieve Frobenius accuracy
 \textup{\begin{equation}
 \mbox{Pr}[\|\hat{\Phi}_{PLS}-\Phi\|_2\:\geq\epsilon\:]\:\leq\:\eta,
 \end{equation}}
 one requires a sample size of
 \textup{\begin{equation}\label{sample sizes}
 N(k,r)\geq\frac{32 r}{g(k)}\frac{8}{3\epsilon^2}\mbox{log}(\frac{2^{2k}}{\eta}).
 \end{equation}}
 Up to logarithmic factors, this results in sampling complexity
 \textup{\begin{equation}
 \begin{aligned}
 N(k,r)=
 \begin{cases}
 \bigo(\frac{1}{\epsilon^2}d^{3.17}r)&\mbox{for Scenarios \ref{sub:pauli-ancilla}, \ref{sub:pauli-direct},}\\
 \bigo(\frac{1}{\epsilon^2}d^2r)&\mbox{for Scenarios \ref{sub:MUB+ancilla}, \ref{sub:MUB-direct}.}
 \end{cases}
 \end{aligned}
 \end{equation}}
 \end{corollary}


These observations showcase that the PLS method for QPT yields sampling complexities which increase linearly for the Frobenius error 
(quadratically for norm-one error) with the rank of the channel, so that low rank channels have lower errors than full rank ones. 
This behavior is very similar to compressed sensing estimators that are designed to exploit (approximate) low rank \cite{Flammia_2012,Roth2018,Kliesch2019}. In contrast, the LS estimator shows weak dependence on rank and for low rank states, its Frobenius error is $O(d^2)$ larger that that of PLS, as shown in Corollary \ref{cor:LS_errors} in Appendix \ref{sec:summary_LS} .

\vspace{2mm}

Although the concentration bounds characterise the PLS error behaviour in terms of rank, it is unlikely that an experimenter knows the rank of the Choi matrix of the quantum process under investigation. Therefore, the bounds cannot be used to construct \emph{confidence regions} unless the rank is known. We now state a more general result which allows us to define confidence regions without prior knowledge of the rank of $\Phi$, but rather in terms of properties of the estimator $\hat{\Phi}_{CP1}$.

To state the result, we first formalise a notion of being close to rank $r$:
\begin{definition}
A state $\rho$ is $\delta$-almost rank $r$ if there is a rank $r$ state $\rho_{(r)}$ such that
\begin{align}
    \| \rho - \rho_{(r)} \|_{\infty} & \leq \delta.
\end{align}
\end{definition}
Theorem \ref{thm:projected_bounds_general} below provides confidence balls for the PLS estimator for both the Frobenius and the trace-norm distance, in terms 
of computable properties of the intermediary estimator $\hat{\Phi}_{CP1}$.


\begin{theorem}\label{thm:projected_bounds_general}
Let $\Phi$ be the Choi matrix of a channel and $\hat{\Phi}_{PLS}$ our estimator, generated with any `projections' satisfying Properties \eqref{prop_CP1} and \eqref{decrease_dist}.
Suppose that $\hat{\Phi}_{CP1}$ 
is $\delta$-almost rank $r$ for some $(r,\delta)$.
Then, with probability 
(at least) $1-\epsilon$,
\begin{equation}\label{Frob_concl1}
   \|\hat{\Phi}_{PLS}-\Phi\|_2^2 \leq 2r \left(\delta + 2\sqrt{\frac{8 \ln(2^{2k}/\epsilon)}{3 N g(k)}}\right)^2.
\end{equation}
The trace distance error $\|\hat{\Phi}_{PLS}-\Phi\|_1$ is instead bounded by

\begin{equation}\label{L1_concl_PhiCP1}
    r \left((4\sqrt{2} + 2)\delta + (4 + 8\sqrt{2})\sqrt{\frac{8 \ln(2^{2k}/\epsilon)}{3 N g(k)}}\right).
\end{equation}
In both cases, $g(k)$ has been defined in \eqref{eq.def.g(k)}.
\end{theorem}\par
\hspace{1cm}
\hfill

The theorem can be applied by choosing the pair $(r,\delta)$ which provides the tightest bound and constitutes the confidence ball. 
All the results of this section follow from Theorem \ref{thm:projected_bounds_complete_general} - proven in appendix \ref{sec:projected_bounds}.

\par We finally note that the constants given in Theorems \ref{thm:projected_bounds} and \ref{thm:projected_bounds_general} and Corollary \ref{thm:sampling_complextites} are likely to the pessimistic and the true error probabilities may be significantly smaller. However, the advantage is that the bounds  hold exactly for any channel and any $N$. It is certainly possible to get tighter asymptotic bounds for large $N$.

\section{Numerical experiments}\label{sec:numerical_expts}

In this section we outline the implementation and results of several computer simulations to illustrate the proposed methods.  

All experiments presented in this section were run on a quad-core Intel Core i3-8100 with 16 GB memory (computation times given in Figure \ref{fig:dimension}). Mainly because of memory requirements, the experiments on 7 qubits in Figure \ref{fig:vignette} were done on 64 cores of Intel Xeon  E7-8890 with 512 GB memory, needing around three days for each.

The code is available on GitHub \cite{HIP_repo}. The repository also contains up-to-date information on the future article on the theoretical underpinnings and heuristics behind HIP \cite{futureHIP}.

\subsection{Thresholded projection on trace-one CP operators}\label{sub:CPthres}

As described in Section \ref{sec:methods_of_projection}, the first step of the PLS method is the projection of $\hat{\Phi}_{LS}$ onto the space of quantum states $\CP1$. This is implemented using the following algorithm which describes a general thresholded projection process.

\begin{algorithm}\caption{
Thresholded Projection on trace-one $\CP$}\label{proj_CP}
\begin{algorithmic}[1]
\Function{projCP}{$\Phi$ matrix with trace $1$, $\tau=$threshold}
\State $\lambda_1 \leq \dots \leq \lambda_d  \gets$ \mbox{eigenvalues of $\Phi$}
\State $v_1, \dots, v_n \gets$ corresponding eigenvectors of $\Phi$
\For{$1 \leq i \leq d$}
\If{$\lambda_i \leq \tau$}
\State $\mu_i \gets 0$
\Else
\State $\mu_i \gets \lambda_i + \tau$
\EndIf \label{lineif}
\EndFor
\If{$\sum \mu_i \geq 1$} \label{condplus}
\State \mbox{Find $x_0$ such that $\sum (\mu_i - x_0)_+ = 1$}
\State $\forall i,\ \mu_i \gets  (\mu_i - x_0)_+$ 
\Else
\State \mbox{find $j$ such that}
\State \quad $\sum_{j+1}^d \lambda_i + (d - j - 1) \tau < 1 \leq \sum_{j}^d \lambda_i + (d - j) \tau$ \label{findj}
\State $\forall i \geq j,\ \mu_i \gets \lambda_i + \tau$
\State $\mu_{j - 1} \gets 1 - \sum_j^d \lambda_i + (d - j) \tau  $
\State $\forall i < j - 1,\ \mu_i \gets 0$
\EndIf
\State \Return 
operator $\sum_i \mu_i |v_i \rangle \! \langle v_i|$
\EndFunction
\end{algorithmic}
\end{algorithm}

The algorithm is separated into two parts. The first, up to and including line 8, is a thresholded projection onto $\CP$, while the second (lines 9 to 17) ensures unit trace and completes the thresholded projection onto $\CP1$.
Note that performing the first part of the algorithm while setting $\tau=0$ produces the direct Frobenius projection onto $\CP$ (see appendix \ref{sec:cp_projection} for the proof). We choose, for projection onto $CP1$,  $\tau=-\lambda_{\min} (\hat{\Phi}_{LS})$ (the sign-flipped, least eigenvalue of $\hat{\Phi}_{LS}$), though we note that the output satisfies Property~\eqref{prop_CP1} regardless of the value of $\tau$ (Lemma \ref{cp1_ls_op_norm_ineq} in the appendix \ref{appndix.proj.prop}).\par
For numeric implementations, we use the full eigenvalue decomposition (EVD) but acknowledge that iterative, randomised implementations or the use of partial EVDs can have a lower theoretical complexity with lower memory requirements - especially in the estimation of high-dimensional, low-rank processes.


\subsection{Hyperplane intersection projection}\label{sub:HIP}

We now move on to discuss the numerical implementation of the second step in the PLS projection, that of $\hat{\Phi}_{CP1}$ onto $\CPTP$. In order to solve this problem, we have developed a fast, generalised `projection' algorithm to find a point in the intersection of an affine space and a convex set satisfying Property~\eqref{decrease_dist}.
Here we give a high-level description of the \emph{hyperplane intersection projection} and prove that it satisfies Property~\eqref{decrease_dist} (cf. Lemma \ref{lem:decrease_dist} in appendix \ref{appndix.proj.prop}).

We know how to efficiently project (in Euclidean distance) onto $\CP$ and $\TP$, though the projection $\texttt{proj}_{\CP}$ on $\CP$ is more costly. (Closed form expressions for these projections are given in appendix \ref{app:details_on_num_expts}, and we do not provide explicit algorithms for them, given their simplicity.)

In the hyperplane intersection projection (HIP) algorithm, we switch between the AP regime (projecting alternately onto $\CP$ and $\TP$) and \emph{HIP mode}, based on a criterion discussed below.
In HIP mode, we keep a list of half-spaces that contain $\CP$ - defined by the estimator after each iteration - and project on the intersection of a large subset of these half-spaces and $\TP$. The key idea that allows efficient computations 
is that it is easy to do the following two things: project onto the intersection of hyperplanes, and check that the projection onto the intersection of hyperplanes is the same as the projection onto the intersection of half-spaces. We select a large subset of half-spaces that ensures that this equality holds, before projecting onto the intersection of $\TP$ and the associated set of hyperplanes. See algorithms \ref{HIP_intern} and \ref{HIP} for pseudocode respectively showing construction of the set of hyperplanes, and the HIP algorithm.

\begin{algorithm}
\caption{Choosing hyperplanes}\label{HIP_intern}
\begin{algorithmic}[1]
\Function{HIP\_inner}{\texttt{w\_list} list of half-spaces, $\Phi$ state}
\State \texttt{w\_active } $\gets$ empty list
\For{\texttt{w} $\in$ \texttt{w\_list}}
\State \Comment{The following test can be checked by looking at whether some coefficients are all non-negative.}
\If{the projection of $\Phi$ on the intersection of $\TP$ and the half-spaces of \texttt{w\_active} and \texttt{w} is equal to the projection of $\Phi$ on
the intersection of $\TP$ and the corresponding hyperplanes}
\State Append \texttt{w} to \texttt{w\_active}
\EndIf
\EndFor
\State {$\Phi_{new} \gets $ the projection of $\Phi$ on the intersection  of $\TP$ and the half-spaces in \texttt{w\_active}}
\State \Return \texttt{w\_active}, $\Phi_{new}$
\EndFunction
\end{algorithmic}
\end{algorithm}

\begin{algorithm}
\caption{Hyperplane Intersection Projection}\label{HIP}
\begin{algorithmic}[1]
\Function{HIP}{$\Phi$ initial state, $\epsilon$ tolerance}
\State \texttt{mode} $\gets$ \texttt{AP}
\While{$\lambda_{\min}(\Phi) < -\epsilon$}
\While{\texttt{mode} $=$ \texttt{AP}}
\State $\Phi \gets \texttt{proj}_{\CP}(\Phi)$
\State $\Phi \gets \texttt{proj}_{\TP}(\Phi)$
\If{SwitchCondition\_toHIP}
\State $\texttt{mode}\gets\texttt{HIP}$
\State $\texttt{w\_active} \gets \textrm{empty list}$
\EndIf
\EndWhile
\While{\texttt{mode} $=$ \texttt{HIP}}
\State $\Phi_{\CP} \gets \texttt{proj}_{\CP}(\Phi)$
\State $\texttt{w} \gets $ \textrm{half-space containing $\CP$ orthogonal at $\Phi_{\CP}$ to $\Phi_{\CP} - \Phi$}
\State Add $\texttt{w}$ as first element of $\texttt{w\_active}$
\State $\texttt{w\_active}, \Phi \gets \texttt{HIP\_inner}(\texttt{w\_active}, \Phi)$
\If{SwitchCondition\_toAP}
\State $\texttt{mode}\gets\texttt{AP}$
\EndIf
\EndWhile
\EndWhile
\State \Return $\Phi$
\EndFunction
\end{algorithmic}
\end{algorithm}

Whatever the criterion for switching between AP and HIP mode, Property~\ref{decrease_dist} is always satisfied. An easy choice that works well in practice is to take a small fixed number of steps (say six) in AP mode, and a bigger fixed number of steps (say thirty) in HIP mode. The criterion we use in the experiments is more convoluted, and can be read in the implementation \cite{HIP_repo}.\\\par

\subsection{Comparison of projection algorithms}\label{sec:algo_comp}


Here, we compare the performance of several variants of the HIP algorithm discussed in Section \ref{sub:HIP}, with AP and Dykstra's algorithm.\\\par

It turns out that the largest reduction in error during the projection of $\hat{\Phi}_{LS}$ onto $\CPTP$ occurs during the initial thresholded projection onto $\CP1$ discussed in Section \ref{sub:CPthres}. After near-convergence of any of the subsequent algorithms discussed, before mixing with the depolarizing channel, the $L^1$ distance $\|\hat{\Phi}_{PLS}-\Phi\|_{1}$ is the same as that of $\|\hat{\Phi}_{CP1}-\Phi\|_{1}$, up to a relative change of five percent. Hence the important question regarding the choice of algorithm is \textit{which converges the fastest}?\par
We test convergence times for `projections' of $\hat{\Phi}_{LS}$ for a 5-qubit quantum Fourier transform, with 5 different algorithmic implementations: AP, Dykstra, HIPswitch, OneHIP (HIP without any memory of past hyperplanes or switches to AP), and PureHIP (no switch to AP). (HIPswitch, OneHIP and PureHIP are variations on Algorithm \ref{HIP}, and we provide their explicit algorithmic implementations in Appendix \ref{sec:additional_proj_algs}.) We do so primarily to test the convergence time of our main algorithm HIPswitch compared to that of AP and Dykstra. Additionally, by turning on parts of HIPswitch in turn (OneHIP, PureHIP) we explore the efficacy of the algorithm. Note also in Figure \ref{fig:algo_choice}, that we display the results of a dual-approach to the projection of $\CPone$ onto $\CPTP$; we defer a discussion of this to Section \ref{sec:dual_opt}.


\begin{figure}[ht]
 \begin{center}
   \includegraphics[width=.65\textwidth]{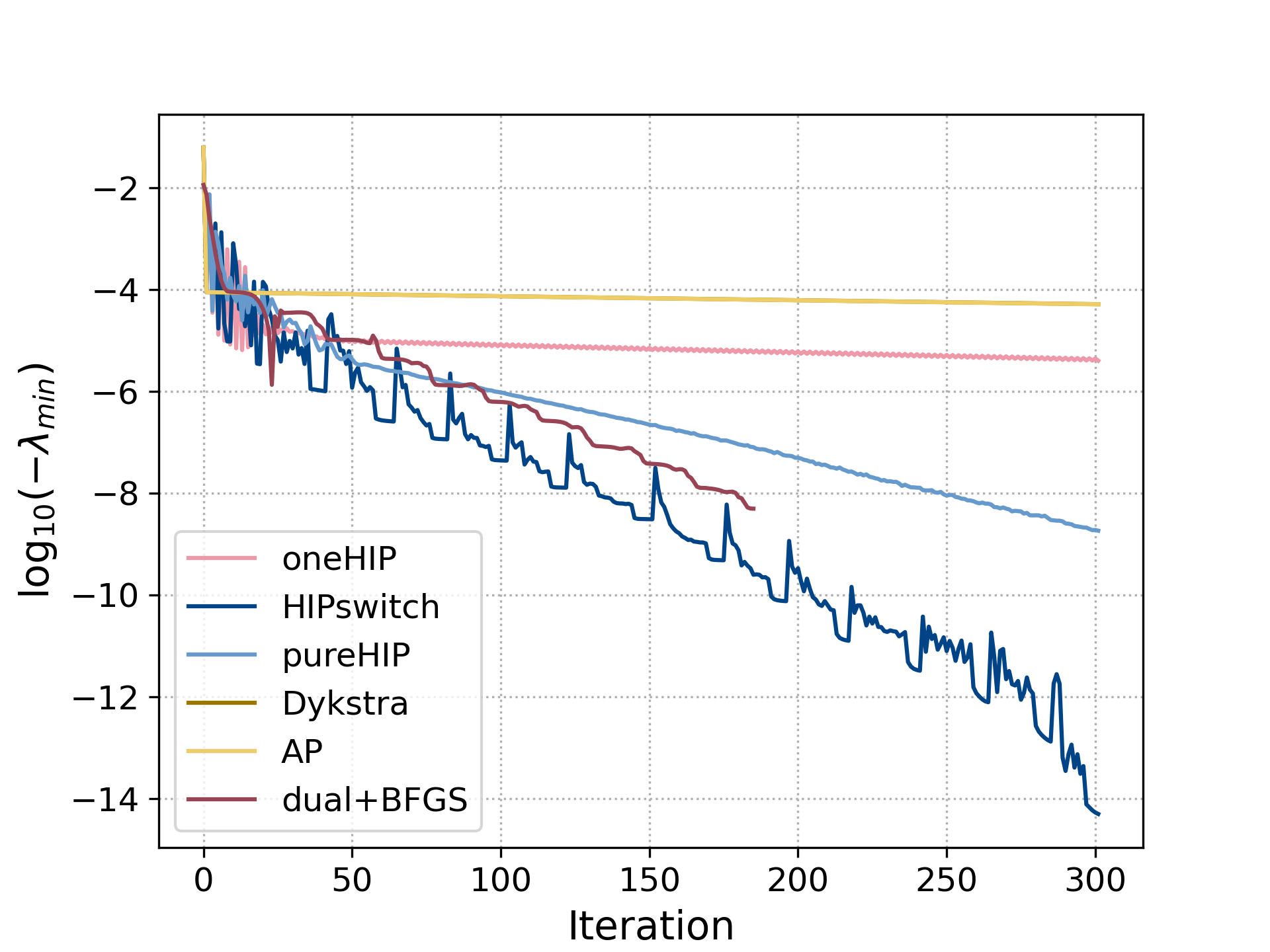}
 \end{center}
 \caption{Convergence of different projections for the 
 quantum Fourier transform on 5 qubits, sample size $10^7$. The plot shows minus the log of the least eigenvalue of $\hat{\Phi}_{\TP}$ as function of number of iterations for different implementations. The projection algorithm based on the dual problem is discussed in section \ref{sec:dual_opt}.}\label{fig:algo_choice}
 \end{figure}
 
 In Figure \ref{fig:algo_choice}, we see that:
 \begin{itemize}
     \item Dykstra's algorithm and AP are almost indistinguishable, and do not converge in reasonable time.
     \item The simplified version oneHIP of HIP, without memory of past hyperplanes or AP, does a little better, but also fails to converge rapidly.
     \item The simplified version of HIP `pureHIP', in which we always stay in HIP mode, improves linearly in this setting.
     \item The main algorithm HIPswitch converges faster than pureHIP. The flat low points correspond to the AP mode. At the end, we might have hit the quadratic convergence regime. HIPswitch and pureHIP attain this faster on easier problems (lower dimension, higher rank) for convergence.
 \end{itemize}

In all subsequent experiments discussed in this article, we use the algorithm HIPswitch (Algorithm \ref{HIP}). We stop when the effect of adding the depolarizing channel changes the trace-norm distance by at most $2\times 10^{-3}$, corresponding to a least eigenvalue of $10^{-7}$ for 5 qubits, i.e. $-7$ on Figure \ref{fig:algo_choice}.

\subsection{Effect of sample size}


\begin{figure}[ht]
 \begin{center}
   \includegraphics[width=.65\textwidth]{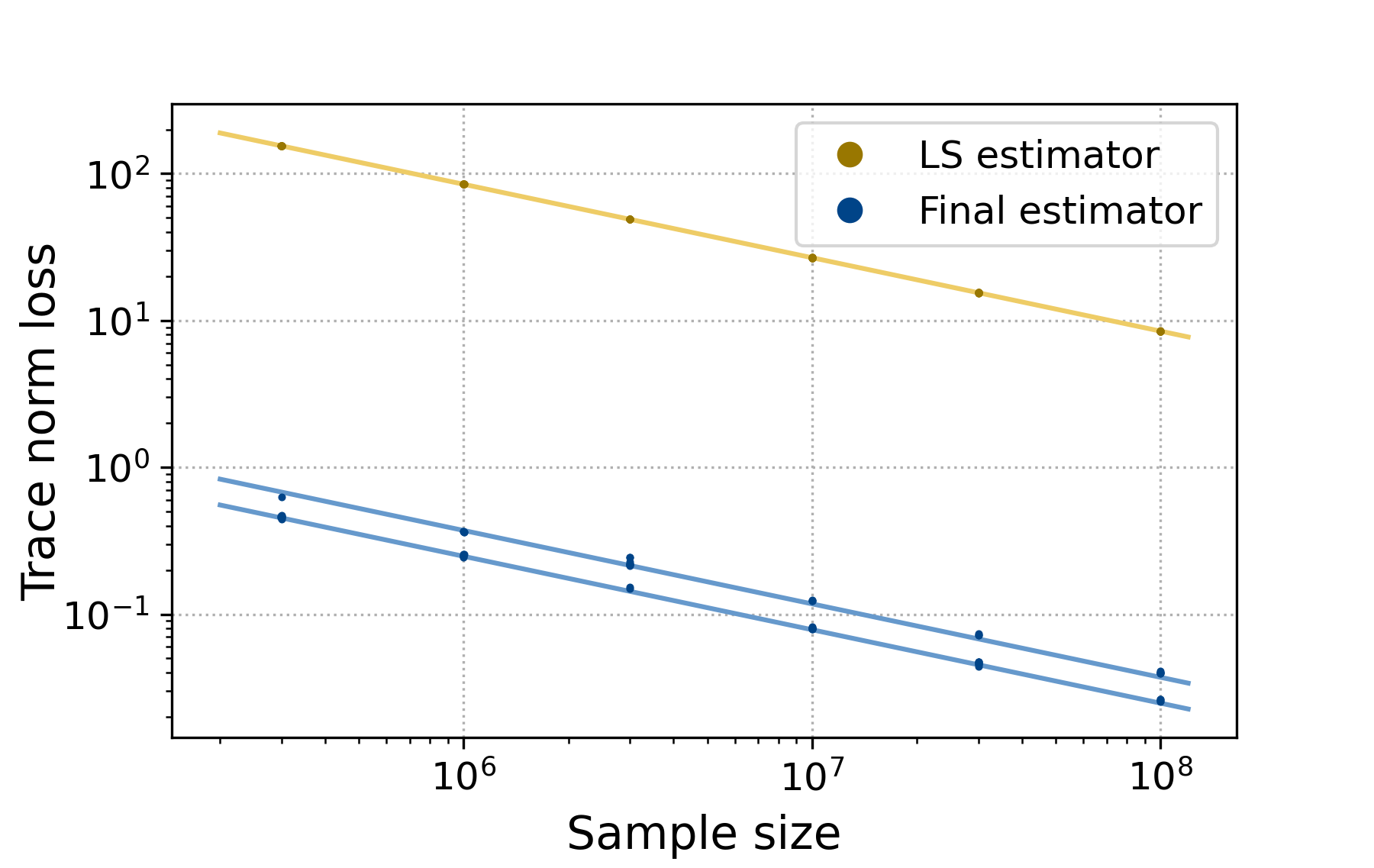}
 \end{center}
 \caption{Trace-norm error as function of sample size for PLS (blue lines) and LS (yellow line) of a QFT on 5 qubits. The slope of the lines  corresponds to a $\sqrt{N}$ rate. The two blue lines correspond to two regimes determined by the rank of the first projection $\hat{\Phi}_{CP1}$, and are roughly $4/d^2$ times the least-squares line.
 }\label{fig:sample_size}
 \end{figure}

We now look at how the performance of our QPT method depends on the number of measurements made. We repeat ten times the estimation of the QFT on 5 qubits, for sample sizes  of $N=3e5, 1e6, 3e6, 1e7, 3e7, 1e8$. We can see on Figure \ref{fig:sample_size} that:
\begin{itemize}
    \item unsurprisingly, the least-squares estimator error scales as $N^{-1/2}$.
    \item The least-squares estimator always has almost the same error for the same sample size: each set of ten points is indistinguishable on the figure.
    \item Our final estimator has two possible regimes. It depends on the rank of the first projection $\hat{\Phi}_{CP1}$. For a given rank, there is visible but low variation of the loss.
    \item Most importantly, the final estimator also scales as $N^{-1/2}$.
    \item For this rank-one channel, the final estimator divides the loss of the least-squares estimator by $d^2/3$ or $d^2/5$. That is, for 5 qubits, we improve by a factor of $300$.
\end{itemize}

\subsection{Effect of rank}

\label{subsec:rank.numerics}
\begin{figure}[ht]
 \begin{center}
   \includegraphics[width=.65\textwidth]{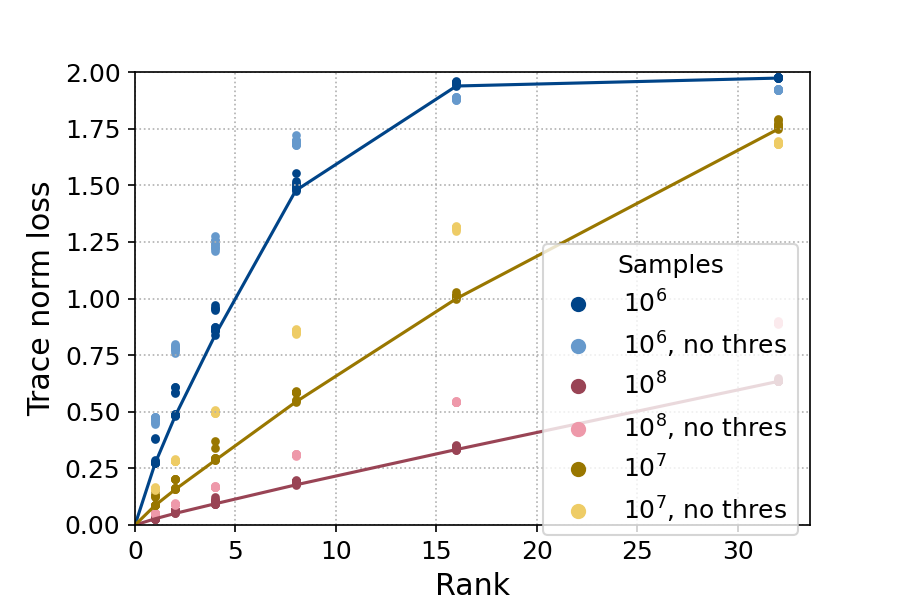}
 \end{center}
 \caption{Trace norm error for 5-qubit channels of different ranks, with different sample sizes. The dots corrsespond to ten repetitions for each experiment. Light colored dots correspond to intial  projection on $\CP1$ with no thresholding. }\label{fig:rank}
 \end{figure}

In this section, we estimate sums of orthogonal unitary channels, with the ranks $1, 2, 4, 8, 16, 32$, and equal eigenvalues. We include a comparison of the performance of the projected estimator when using the two methods of initial projection onto $\CP1$: our default thresholded projection, and Frobenius projection, using the same data. We test for sample sizes $N=1e6, 1e7, 1e8$. Each experiment is reproduced ten times. In Figure \ref{fig:rank}, we see that:
\begin{itemize}
    \item the loss is slightly less than linear in the rank, until it is close to the maximum of $2$.
    \item The thresholded projection is much better than the standard projection, unless the loss is already close to the maximum of $2$, in which case it is worse.
    \item The effect of getting the wrong rank with the thresholded projection decreases with the rank.
\end{itemize}

\subsection{Effect of dimension}

We now test the estimation of the Quantum Fourier Transform using Pauli measurements on one to six qubits, and for mutually unbiased measurements in dimensions $3, 7, 11, 17, 31, 67$. To get comparable losses, we copy the scaling of Theorem \ref{thm:projected_bounds}, and set the respective sample sizes as $10 \times 9^k$ for $k$ qubits, and $100 \times d^2$, equivalent to $100 \times 4^k$, for MUBs.

\begin{figure}[ht]
 \begin{center}
   \includegraphics[width=.65\textwidth]{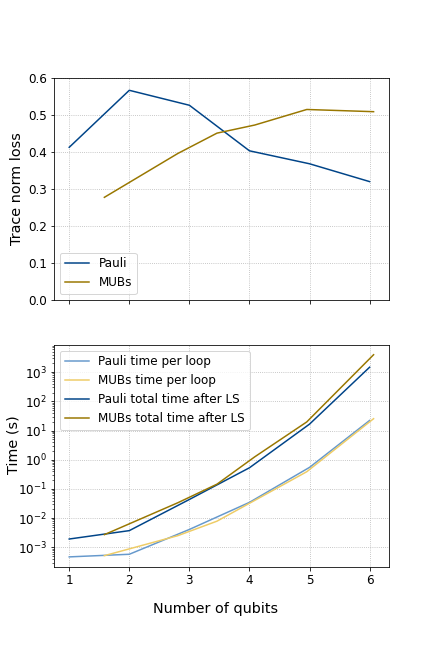}
 \end{center}
 \caption{Sample size: $10 \times 9^k$ for Pauli, $100 \times 4^k$ for MUBs.}\label{fig:dimension}
 \end{figure}

  

In Figure \ref{fig:dimension}, we show the mean trace norm loss and time requirements over ten experiments. 
For the latter, we show the time needed for one loop of algorithm \ref{HIP}, and the time needed for the whole estimation. (We do not include the time taken for the generation of the least-squares estimator, as it is included in the data generation.) We can see that:
\begin{itemize}
    \item the scaling of Theorem \ref{thm:projected_bounds} seems comparable to what we get in practice.
    \item The time for each loop does not depend on the setting. Slightly more loops are usually needed for MUBs than for Pauli.
    \item We see that the scaling in time gets steeper: possibly for low dimension, bigger arrays are computed faster per entry. The most time-consuming operation is the diagonalisation in $\mbox{proj}_{CP}$, and its implementation will dictate the complexity. Going from five to six qubits makes loops forty times slower, and the whole computation ninety times slower, hence more loops are required.
    \item Experiments on six qubits can be done in half an hour on a personal computer (plus $15$ minutes to generate the data); on 5 qubits in less than twenty seconds.
\end{itemize}
 
The memory requirements are at the very minimum of $d^4$, the size of the Choi matrix. With Algorithm 
\ref{HIP}, we need to keep one such matrix in memory for each hyperplane, which can multiply the needs by $20$. However, the limiting factor for the simulations in the Pauli setting is the data generation, where a careful reasonably fast implementation requires $24^k$ complex entries.

\subsection{An alternative optimisation procedure}\label{sec:dual_opt}

We thank an anonymous referee for pointing out that the dual optimisation problem for finding the projection onto $\mathcal{CPTP}$ may be more well-behaved than the primal one. Starting from that point, it is possible to devise a simpler projection algorithm that yields the Euclidean projection, is well-rooted in optimisation theory, and has speeds comparable to HIP - based on the few experiments we tested it on. In this section we outline the method and present a preliminary result but we postpone a more in-depth analysis and comparison to HIP for the future publication \cite{futureHIP}.

The direct projection problem (without projection of $\LSE$ onto $\CP1$) may be written as follows:
\begin{equation}\label{primal}
    \PLS = \argmin_{\Phi \in \CP, \Phi \in \TP} \left\lVert \Phi - \LSE \right\rVert_2.
\end{equation}

The constraint $\Phi \in \CP$ can be made explicit as $\Phi \geq 0$, that is, a set of linear inequality constraints, and the constraint $\Phi \in \TP$ can be made explicit as $\Tr_s(\Phi) = \frac1d \mathbb{1}_d$, that is, a set of linear equality constraints. Hence, we can look at the following dual problem. We relax the trace-preserving constraint and form the Lagrangian:
\begin{equation}
\mathcal{L}(\Phi, \nu) = \|\Phi-\LSE\|_2 + \langle \nu \vert \Tr_s(\Phi) - \frac1d \mathbb{1}_{d}\rangle_{F},
\end{equation}
with $\nu \in \mbox{Herm}_{n\times n}$ our dual variable.
Let the solution to this relaxed problem be:
\begin{equation}
    \Phi_{rel}(\nu) = \argmin_{\Phi\in\CP}\mathcal{L}(\Phi, \nu).
\end{equation}
The associated dual function is
\begin{equation}
    q(\nu) = \min_{\Phi\in\CP}\mathcal{L}(\Phi, \nu),
\end{equation}
and the optimal value of $\nu$ is found by solving the dual problem:
\begin{equation}
    \nu_{opt} = \argmax_{\nu\in\Herm_{n\times n}}q(\nu).
\end{equation}

The totally mixed state $\Phi = \frac1{d^2} \mathbb{1}_{d^2}$ strictly satisfies the inequalities, hence Slater's conditions \cite{Slater, Boyd} hold, and the solution of the primal problem \eqref{primal} is the same as the one yielded by the dual problem: $\Phi_{rel}(\nu_{opt})=\PLS$.\par

This is beneficial as the solution to the relaxed problem is easy to compute:
\begin{align*}
    \mathcal{L}(\Phi, \nu) & = \left\lVert \Phi - \LSE \right\rVert_2 + \langle \nu \vert \Tr_s(\Phi) - \frac1d \mathbb{1}_{d}\rangle_{F} \\
    & = \left\lVert \Phi - \LSE \right\rVert_2 + \langle \mathbb{1}_d \otimes \nu \vert \Phi - \frac1{d^2} \mathbb{1}_{d^2}\rangle_{F} \\
    & = \left\lVert \Phi - (\LSE -  \frac12 \mathbb{1}_d \otimes \nu) \right\rVert_2 - f(\LSE, \nu),
\end{align*}
for a function $f$ which we do not need to make explicit, so that
\begin{align}
    \label{Phinu}
    \Phi_{rel}(\nu) & = \mathrm{proj}_{\mathcal{CP}}(\LSE -  \frac12 \mathbb{1}_d \otimes \nu).
\end{align}
This allows simple computation of the Lagrangian. We can also check that the gradient is $\nabla q(\nu) = \Tr_s(\Phi(\nu)) - \frac1d \mathbb{1}_{d}$.\par

This dual problem is well-behaved: to start with, there are only $d^2$ real parameters, instead of $d^4$. Moreover, the function is much smoother: the gradient is Lipschitz, so that we can use fast optimisation algorithms. We used BFGS (\cite{broyden1970convergence, shanno1970conditioning, fletcher1970new, goldfarb1970family}), since it is available in \textrm{scipy}, but options such as Nesterov's accelerated gradient \cite{nesterov1983method} could be considered. 

Some simple simulations show that the speeds obtained using this algorithm are comparable with HIP (Figure \ref{fig:algo_choice}), with again most of the time spent in $\mathrm{proj}_{\mathcal{CP}}$. We need slightly more loops than for HIP in the five qubit case, as in Figure \ref{fig:algo_choice}, but there is likely to be room for improvement. Notice that the line corresponding to the dual+BFGS algorithm stops around $10^{-8}$. This is because this implementation of the optimisation algorithm stops for lack of precision at that point.

\par

We can also give an interesting interpretation of the dual variable $\nu$: when we use alternate projections, we project on $\mathcal{CP}$, then correct by projecting on $\mathcal{TP}$, and so on. By contrast, in the dual optimisation, we look for an element in the orthogonal space to $\mathcal{TP}$ to be added to $\LSE$, so that the projection on $\mathrm{proj}_{\mathcal{CP}}$ belongs to $\mathcal{TP}$, too.

\section{Conclusions and outlook}
In this work we proposed and investigated a computationally efficient and statistically tractable quantum process tomography method called projected least squares (PLS), inspired by previous work on quantum state tomography \cite{Gu__2020}. The estimator is constructed by first computing the LS estimator and then projecting this on the space of physical Choi matrices. 
We propose this projection be carried out in two steps: a projection onto the set of quantum states followed by a projection onto the set of Choi matrices. Additionally, we proposed a fast numerical implementation of this projection -- hyperplane intersection projection -- which accelerates an alternating projections procedure by employing fast projections onto the intersection of several tangent spaces and the hyperplane of linear constraints. 

We studied four experimental scenarios for data generation. We consider the most experimentally feasible to be ancilla-free tomography with single-qubit Pauli measurements (scenario \ref{sub:pauli-direct}), while the most statistically efficient is the ancilla-assisted setup with mutually unbiased bases measurements (scenario \ref{sub:MUB+ancilla}). For each scenario we provided explicit expressions for the LS estimator (Section~\ref{sec:experimental_procedures}), and non-asymptotic concentration bounds for the PLS estimator with respect to the Frobenius and trace norm error of the associated Choi matrix, cf. Theorem \ref{thm:projected_bounds}.


Via the two-step projection procedure, all proposed methods are able to exploit latent structure in the form of low rank.
Low rank channels are particularly interesting as a class of transformations modelling noisy implementations of unitary channels, which are relevant for quantum technology and quantum computing. Our sampling complexity bounds for the Frobenius error depend linearly on the rank $r$ of the channel and are smaller than those of the LS estimator by a factor $r/d^2$. This qualitative behaviour is confirmed by numerical simulations which show that the PLS errors are significantly lower than those of LS for low rank channels. 

We carried out several numerical studies involving 5 to 7 qubit channels, including a comparison between the two projection methods, a study on the error dependence on rank, and an analysis of how the computational time scales with system dimension.  
The code is available on GitHub \cite{HIP_repo}. The repository also contains up-to-date information on the future article on the theoretical underpinnings and heuristics behind HIP \cite{futureHIP}. Following a suggestion by an anonymous referee, we implemented an alternative projection algorithm based on the dual optimisation problem, which shows a similar behaviour to HIP and will be analysed in more detail in \cite{futureHIP}. 

We leave as a topic of future investigations the extension of these results to quantum instruments described in terms of non trace-preserving quantum operations. In a different direction, we would like to further investigate the efficacy and possible improvement of the confidence regions prescribed by Theorem \ref{thm:projected_bounds_general}. Research should also be carried out to find improvements to the algorithmic implementation of the projections of $\LSE$ onto $\CPTP$ to decrease computational complexity and memory demands. These could arise from ``matrix-free'' representations of the least-squares estimator and from exploiting the low rank of many high-dimensional processes of interest \cite{doi:10.1137/19M1305045}.


\par
 

\textbf{Acknowledgements:}
RK would like to thank Ingo Roth for inspiring discussions during the early stages of this project. For the figures we used the colour-blind-safe colour scheme from \cite{Tol_colours}.


\bibliographystyle{quantum}
\bibliography{plsqpt_article_refs}

\begin{thebibliography}{10}

\bibitem{Riebe2006}
M.~Riebe, K.~Kim, P.~Schindler, T.~Monz, P.~O. Schmidt, T.~K. K\"{o}rber,
  W.~H\"{a}nsel, H.~H\"{a}ffner, C.~F. Roos, and R.~Blatt.
\newblock ``Process tomography of ion trap quantum gates''.
\newblock
  \href{https://dx.doi.org/https://doi.org/10.1103/PhysRevLett.97.220407}{Phys.
  Rev. Lett. {\bf 97}, 220407}~(2006).

\bibitem{Weinstein2004}
Y.~S. Weinstein, T.~F. Havel, J.~Emerson, N.~Boulant, M.~Saraceno, S.~Lloyd,
  and D.~G. Cory.
\newblock ``Quantum process tomography of the quantum fourier transform''.
\newblock \href{https://dx.doi.org/https://doi.org/10.1063/1.1785151}{The
  Journal of Chemical Physics {\bf 121}, 6117–6133}~(2004).

\bibitem{PhysRevLett.93.080502}
J.~L. O'Brien, G.~J. Pryde, A.~Gilchrist, D.~F.~V. James, N.~K. Langford, T.~C.
  Ralph, and A.~G. White.
\newblock ``Quantum process tomography of a controlled-not gate''.
\newblock \href{https://dx.doi.org/10.1103/PhysRevLett.93.080502}{Phys. Rev.
  Lett. {\bf 93}, 080502}~(2004).

\bibitem{Pach_n_2015}
L.~A. Pachón, A.~H. Marcus, and A.~Aspuru-Guzik.
\newblock ``Quantum process tomography by 2d fluorescence spectroscopy''.
\newblock \href{https://dx.doi.org/10.1063/1.4919954}{The Journal of Chemical
  Physics {\bf 142}, 212442}~(2015).

\bibitem{Bialczak_2010}
R.~C. Bialczak, M.~Ansmann, M.~Hofheinz, E.~Lucero, M.~Neeley, A.~D.
  O’Connell, D.~Sank, H.~Wang, J.~Wenner, M.~Steffen, and et~al.
\newblock ``Quantum process tomography of a universal entangling gate
  implemented with josephson phase qubits''.
\newblock \href{https://dx.doi.org/10.1038/nphys1639}{Nature Physics {\bf 6},
  409–413}~(2010).

\bibitem{Howard_2006}
M.~Howard, J.~Twamley, C.~Wittmann, T.~Gaebel, F.~Jelezko, and J.~Wrachtrup.
\newblock ``Quantum process tomography and linblad estimation of a solid-state
  qubit''.
\newblock \href{https://dx.doi.org/10.1088/1367-2630/8/3/033}{New Journal of
  Physics {\bf 8}, 33–33}~(2006).

\bibitem{Chuang_1997}
I.~L. Chuang and M.~A. Nielsen.
\newblock ``Prescription for experimental determination of the dynamics of a
  quantum black box''.
\newblock \href{https://dx.doi.org/10.1080/09500349708231894}{Journal of Modern
  Optics {\bf 44}, 2455–2467}~(1997).

\bibitem{PhysRevLett.78.390}
J.~F. Poyatos, J.~I. Cirac, and P.~Zoller.
\newblock ``Complete characterization of a quantum process: The two-bit quantum
  gate''.
\newblock \href{https://dx.doi.org/10.1103/PhysRevLett.78.390}{Phys. Rev. Lett.
  {\bf 78}, 390--393}~(1997).

\bibitem{CHOI1975285}
M.~D. Choi.
\newblock ``Completely positive linear maps on complex matrices''.
\newblock
  \href{https://dx.doi.org/https://doi.org/10.1016/0024-3795(75)90075-0}{Linear
  Algebra and its Applications {\bf 10}, 285 -- 290}~(1975).

\bibitem{JAMIOLKOWSKI1972275}
A.~Jamiołkowski.
\newblock ``Linear transformations which preserve trace and positive
  semidefiniteness of operators''.
\newblock
  \href{https://dx.doi.org/https://doi.org/10.1016/0034-4877(72)90011-0}{Reports
  on Mathematical Physics {\bf 3}, 275 -- 278}~(1972).

\bibitem{leung2000robust}
D.~W. Leung.
\newblock ``Towards robust quantum computation''.
\newblock
  \href{https://dx.doi.org/https://doi.org/10.48550/arXiv.cs/0012017}{Thesis
  (PhD). STANFORD UNIVERSITY, Source DAI-B 61/11, p. 5911, 225 pages}~(2000).
\newblock  \href{http://arxiv.org/abs/cs/0012017}{arXiv:cs/0012017}.

\bibitem{PhysRevLett.86.4195}
G.~M. D'Ariano and P.~Lo~Presti.
\newblock ``Quantum tomography for measuring experimentally the matrix elements
  of an arbitrary quantum operation''.
\newblock \href{https://dx.doi.org/10.1103/PhysRevLett.86.4195}{Phys. Rev.
  Lett. {\bf 86}, 4195--4198}~(2001).

\bibitem{PhysRevA.64.052312}
D.~F.~V. James, P.~G. Kwiat, W.~J. Munro, and A.~G. White.
\newblock ``Measurement of qubits''.
\newblock \href{https://dx.doi.org/10.1103/PhysRevA.64.052312}{Phys. Rev. A
  {\bf 64}, 052312}~(2001).

\bibitem{Lvovsky_2004}
A.~I. Lvovsky.
\newblock ``Iterative maximum-likelihood reconstruction in quantum homodyne
  tomography''.
\newblock \href{https://dx.doi.org/10.1088/1464-4266/6/6/014}{Journal of Optics
  B: Quantum and Semiclassical Optics {\bf 6}, S556--S559}~(2004).

\bibitem{PhysRevLett.105.200504}
R.~Blume-Kohout.
\newblock ``Hedged maximum likelihood quantum state estimation''.
\newblock \href{https://dx.doi.org/10.1103/PhysRevLett.105.200504}{Phys. Rev.
  Lett. {\bf 105}, 200504}~(2010).

\bibitem{Smolin_2012}
J.~A. Smolin, J.~M. Gambetta, and G.~Smith.
\newblock ``Efficient method for computing the maximum-likelihood quantum state
  from measurements with additive gaussian noise''.
\newblock \href{https://dx.doi.org/10.1103/PhysRevLett.108.070502}{Phys. Rev.
  Lett. {\bf 108}, 070502}~(2012).

\bibitem{Granade_2016}
L.~Granade, J.~Combes, and D.~G. Cory.
\newblock ``Practical bayesian tomography''.
\newblock \href{https://dx.doi.org/10.1088/1367-2630/18/3/033024}{New Journal
  of Physics {\bf 18}, 033024}~(2016).

\bibitem{PhysRevLett.102.020504}
M.~Christandl, R.~K\"onig, and R.~Renner.
\newblock ``Postselection technique for quantum channels with applications to
  quantum cryptography''.
\newblock \href{https://dx.doi.org/10.1103/PhysRevLett.102.020504}{Phys. Rev.
  Lett. {\bf 102}, 020504}~(2009).

\bibitem{Blume_Kohout_2010}
R.~Blume-Kohout.
\newblock ``Optimal, reliable estimation of quantum states''.
\newblock \href{https://dx.doi.org/10.1088/1367-2630/12/4/043034}{New Journal
  of Physics {\bf 12}, 043034}~(2010).

\bibitem{Granade_2017}
C.~Granade, C.~Ferrie, and S.~T. Flammia.
\newblock ``Practical adaptive quantum tomography''.
\newblock \href{https://dx.doi.org/10.1088/1367-2630/aa8fe6}{New Journal of
  Physics {\bf 19}, 113017}~(2017).

\bibitem{Haffner_2005}
H.~H\"{a}ffner, W.~H\"{a}nsel, C.~F. Roos, J.~Benhelm, D.~Chek-al kar,
  M.~Chwalla, T.~K\"{o}rber, U.~D. Rapol, M.~Riebe, P.~O. Schmidt, C.~Becher,
  O.~G\"{u}hne, W.~D\"{u}r, and R.~Blatt.
\newblock ``Scalable multiparticle entanglement of trapped ions''.
\newblock \href{https://dx.doi.org/https://doi.org/10.1038/nature04279}{Nature
  {\bf 438}, 643–646}~(2005).

\bibitem{Christandl&Renner}
M.~Christandl and R.~Renner.
\newblock ``Reliable quantum state tomography''.
\newblock \href{https://dx.doi.org/10.1103/PhysRevLett.109.120403}{Phys. Rev.
  Lett. {\bf 109}, 120403}~(2012).

\bibitem{Faist&Renner}
P.~Faist and R.~Renner.
\newblock ``Practical and reliable error bars in quantum tomography''.
\newblock \href{https://dx.doi.org/10.1103/PhysRevLett.117.010404}{Phys. Rev.
  Lett. {\bf 117}, 010404}~(2016).

\bibitem{Faist2019}
L.~P. Thinh, P.~Faist, J.~Helsen, D.~Elkouss, and S.~Wehner.
\newblock ``Practical and reliable error bars for quantum process tomography''.
\newblock \href{https://dx.doi.org/10.1103/PhysRevA.99.052311}{Phys. Rev. A
  {\bf 99}, 052311}~(2019).

\bibitem{Flammia_2012}
S.~T. Flammia, D.~Gross, Y.~K. Liu, and J.~Eisert.
\newblock ``Quantum tomography via compressed sensing: error bounds, sample
  complexity and efficient estimators''.
\newblock \href{https://dx.doi.org/10.1088/1367-2630/14/9/095022}{New Journal
  of Physics {\bf 14}, 095022}~(2012).

\bibitem{Roth2018}
I.~Roth, R.~Kueng, S.~Kimmel, Y.-K. Liu, D.~Gross, J.~Eisert, and M.~Kliesch.
\newblock ``Recovering quantum gates from few average gate fidelities''.
\newblock \href{https://dx.doi.org/10.1103/PhysRevLett.121.170502}{Phys. Rev.
  Lett. {\bf 121}, 170502}~(2018).

\bibitem{Kliesch2019}
M.~Kliesch, R.~Kueng, J.~Eisert, and D.~Gross.
\newblock ``Guaranteed recovery of quantum processes from few measurements''.
\newblock \href{https://dx.doi.org/10.22331/q-2019-08-12-171}{{Quantum} {\bf
  3}, 171}~(2019).

\bibitem{Cramer2009}
M.~Cramer, M.~B. Plenio, S.~T. Flammia, R.~Somma, D.~Gross, S.~D. Bartlett,
  O.~Landon-Cardinal, D.~Poulin, and Y.~K. Liu.
\newblock ``Efficient quantum state tomography''.
\newblock \href{https://dx.doi.org/https://doi.org/10.1038/ncomms1147}{Nature
  communications {\bf 1}, 149}~(2009).

\bibitem{Baumgratz2013}
T.~Baumgratz, D.~Gross, M.~Cramer, and M.~B. Plenio.
\newblock ``Scalable reconstruction of density matrices''.
\newblock
  \href{https://dx.doi.org/https://doi.org/10.1103/PhysRevLett.111.020401}{Phys.
  Rev. Lett. {\bf 111}, 020401}~(2013).

\bibitem{Lanyon2017}
B.~P. Lanyon, C.~Maier, M.~Holz\"{a}pfel, T.~Baumgratz, C.~Hempe, P.~Jurcevic,
  I.~Dhand, A.~S. Buyskikh, A.~J. Daley, M.~Cramer, M.~B. Plenio, R.~Blatt, and
  C.~F. Roos.
\newblock ``Efficient tomography of a quantum many-body system''.
\newblock \href{https://dx.doi.org/https://doi.org/10.1038/nphys4244}{Nature
  Physics {\bf 13}, 1158–1162}~(2017).

\bibitem{Torlai2018}
G.~Torlai, G.~Mazzola, J.~Carrasquilla, M.~Troyer, R.~Melko, , and G.~Carleo.
\newblock ``Neural network quantum state tomography''.
\newblock
  \href{https://dx.doi.org/https://doi.org/10.1038/s41567-018-0048-5}{Nature
  Physics {\bf 14}, 447–450}~(2018).

\bibitem{torlai2020}
G.~Torlai, C.~J. Wood, A.~Acharya, G.~Carleo, J.~Carrasquilla, and L.~Aolita.
\newblock ``Quantum process tomography with unsupervised learning and tensor
  networks''~(2020).
\newblock  \href{http://arxiv.org/abs/2006.02424}{arXiv:2006.02424}.

\bibitem{PhysRevLett.102.090502}
J.~M. Chow, J.~M. Gambetta, L.~Tornberg, Jens Koch, Lev~S. Bishop, A.~A. Houck,
  B.~R. Johnson, L.~Frunzio, S.~M. Girvin, and R.~J. Schoelkopf.
\newblock ``Randomized benchmarking and process tomography for gate errors in a
  solid-state qubit''.
\newblock \href{https://dx.doi.org/10.1103/PhysRevLett.102.090502}{Phys. Rev.
  Lett. {\bf 102}, 090502}~(2009).

\bibitem{Knill_2008}
E.~Knill, D.~Leibfried, R.~Reichle, J.~Britton, R.~B. Blakestad, J.~D. Jost,
  C.~Langer, R.~Ozeri, S.~Seidelin, and D.~J. Wineland.
\newblock ``Randomized benchmarking of quantum gates''.
\newblock \href{https://dx.doi.org/10.1103/physreva.77.012307}{Physical Review
  A{\bf 77}}~(2008).

\bibitem{PhysRevLett.108.260506}
L.~Steffen, M.~P. da~Silva, A.~Fedorov, M.~Baur, and A.~Wallraff.
\newblock ``Experimental monte carlo quantum process certification''.
\newblock \href{https://dx.doi.org/10.1103/PhysRevLett.108.260506}{Phys. Rev.
  Lett. {\bf 108}, 260506}~(2012).

\bibitem{PhysRevLett.107.210404}
M.~P. da~Silva, O.~Landon-Cardinal, and D.~Poulin.
\newblock ``Practical characterization of quantum devices without tomography''.
\newblock \href{https://dx.doi.org/10.1103/PhysRevLett.107.210404}{Phys. Rev.
  Lett. {\bf 107}, 210404}~(2011).

\bibitem{Gu__2020}
M.~Gu{\c{t}}{\u{a}}, J.~Kahn, R.~Kueng, and J.~A. Tropp.
\newblock ``Fast state tomography with optimal error bounds''.
\newblock \href{https://dx.doi.org/10.1088/1751-8121/ab8111}{Journal of Physics
  A: Mathematical and Theoretical {\bf 53}, 204001}~(2020).

\bibitem{PhysRevA.98.062336}
G.~C. Knee, E.~Bolduc, J.~Leach, and E.~M. Gauger.
\newblock ``Quantum process tomography via completely positive and
  trace-preserving projection''.
\newblock \href{https://dx.doi.org/10.1103/PhysRevA.98.062336}{Phys. Rev. A
  {\bf 98}, 062336}~(2018).

\bibitem{10.2307/2288193}
R.~L. Dykstra.
\newblock ``An algorithm for restricted least squares regression''.
\newblock Journal of the American Statistical Association {\bf 78},
  837--842~(1983).
\newblock
  url:~\href{http://www.jstor.org/stable/2288193}{http://www.jstor.org/stable/2288193}.

\bibitem{Paris}
M.~Paris and J.~Rehacek, editors.
\newblock ``Quantum state estimation''.
\newblock \href{https://dx.doi.org/https://doi.org/10.1007/b98673}{Volume 649
  of Lecture Notes in Physics}.
\newblock Springer. ~(2004).

\bibitem{sugiyama_least-squares_2013}
T.~Sugiyama, P.~S. Turner, and M.~Murao.
\newblock ``Precision-guaranteed quantum tomography''.
\newblock \href{https://dx.doi.org/10.1103/PhysRevLett.111.160406}{Phys. Rev.
  Lett. {\bf 111}, 160406}~(2013).

\bibitem{kliesch_tutorial_2021}
M.~Kliesch and I.~Roth.
\newblock ``Theory of quantum system certification''.
\newblock \href{https://dx.doi.org/10.1103/PRXQuantum.2.010201}{PRX Quantum
  {\bf 2}, 010201}~(2021).

\bibitem{Kahn_2009}
J.~Kahn and M.~Gu\c{t}\u{a}.
\newblock ``Local asymptotic normality for finite dimensional quantum
  systems''.
\newblock
  \href{https://dx.doi.org/https://doi.org/10.1007/s00220-009-0787-3}{Communications
  in Mathematical Physics {\bf 289}, 597–652}~(2009).

\bibitem{Wright2016}
R.~O'Donnell and J.~Wright.
\newblock ``Efficient quantum tomography''.
\newblock In Proceedings of the Forty-eighth Annual ACM Symposium on Theory of
  Computing.
\newblock \href{https://dx.doi.org/10.1145/2897518.2897544}{Pages 899--912}.
\newblock STOC '16New York, NY, USA~(2016). ACM.

\bibitem{Haah2017}
J.~Haah, A.~W. Harrow, Z.~Ji, X.~Wu, and N.~Yu.
\newblock ``Sample-optimal tomography of quantum states''.
\newblock
  \href{https://dx.doi.org/https://doi.org/10.1109/TIT.2017.2719044}{IEEE T.
  Inform. Theory {\bf 63}, 5628--5641}~(2017).

\bibitem{Giovannetti_2011}
V.~Giovannetti, S.~Lloyd, and L.~Maccone.
\newblock ``Advances in quantum metrology''.
\newblock
  \href{https://dx.doi.org/https://doi.org/10.1038/nphoton.2011.35}{Nature
  Photonics {\bf 5}, 222–229}~(2011).

\bibitem{Zhou_2021}
S.~Zhou and L.~Jiang.
\newblock ``Asymptotic theory of quantum channel estimation''.
\newblock \href{https://dx.doi.org/10.1103/PRXQuantum.2.010343}{PRX Quantum
  {\bf 2}, 010343}~(2021).

\bibitem{Schwinger1960}
J.~Schwinger.
\newblock ``Unitary operator bases''.
\newblock \href{https://dx.doi.org/10.1073/pnas.46.4.570}{Proc. Nat. Acad. Sci.
  U.S.A. {\bf 46}, 570--579}~(1960).

\bibitem{Wootters1989}
W.~K. Wootters and B.~D. Fields.
\newblock ``Optimal state-determination by mutually unbiased measurements''.
\newblock \href{https://dx.doi.org/10.1016/0003-4916(89)90322-9}{Ann. Physics
  {\bf 191}, 363--381}~(1989).

\bibitem{Gross2015}
D.~Gross, F.~Krahmer, and R.~Kueng.
\newblock ``A partial derandomization of phaselift using spherical designs''.
\newblock \href{https://dx.doi.org/10.1007/s00041-014-9361-2}{J. Fourier Anal.
  Appl. {\bf 21}, 229--266}~(2015).

\bibitem{Renes_2004}
J.~M. Renes, R.~Blume-Kohout, A.~J. Scott, and C.~M. Caves.
\newblock ``Symmetric informationally complete quantum measurements''.
\newblock \href{https://dx.doi.org/10.1063/1.1737053}{Journal of Mathematical
  Physics {\bf 45}, 2171–2180}~(2004).

\bibitem{PhysRevA.80.012304}
C.~Dankert, R.~Cleve, J.~Emerson, and E.~Livine.
\newblock ``Exact and approximate unitary 2-designs and their application to
  fidelity estimation''.
\newblock \href{https://dx.doi.org/10.1103/PhysRevA.80.012304}{Phys. Rev. A
  {\bf 80}, 012304}~(2009).

\bibitem{brandao2016local}
F.~G. S.~L. Brand\~{a}o, A.~W. Harrow, and M.~Horodecki.
\newblock ``Local random quantum circuits are approximate polynomial-designs''.
\newblock \href{https://dx.doi.org/10.1007/s00220-016-2706-8}{Comm. Math. Phys.
  {\bf 346}, 397--434}~(2016).

\bibitem{hunterjones2019unitary}
N.~Hunter-Jones.
\newblock ``Unitary designs from statistical mechanics in random quantum
  circuits''~(2019).
\newblock  \href{http://arxiv.org/abs/1905.12053}{arXiv:1905.12053}.

\bibitem{PhysRevA.104.022417}
Jonas Haferkamp and Nicholas Hunter-Jones.
\newblock ``Improved spectral gaps for random quantum circuits: Large local
  dimensions and all-to-all interactions''.
\newblock \href{https://dx.doi.org/10.1103/PhysRevA.104.022417}{Phys. Rev. A
  {\bf 104}, 022417}~(2021).

\bibitem{bauschke2011convex}
H.~H. Bauschke, P.~L. Combettes, et~al.
\newblock ``Convex analysis and monotone operator theory in hilbert spaces''.
\newblock
  \href{https://dx.doi.org/https://doi.org/10.1007/978-1-4419-9467-7}{Volume
  408}.
\newblock Springer. ~(2011).

\bibitem{Gilchrist_2005}
A.~Gilchrist, N.~K. Langford, and M.~A. Nielsen.
\newblock ``Distance measures to compare real and ideal quantum processes''.
\newblock \href{https://dx.doi.org/10.1103/physreva.71.062310}{Physical Review
  A{\bf 71}}~(2005).

\bibitem{watrous_2018}
J.~Watrous.
\newblock ``Unital channels and majorization''.
\newblock \href{https://dx.doi.org/10.1017/9781316848142.005}{Page 201–249}.
\newblock Cambridge University Press. ~(2018).

\bibitem{Wallman2015}
J.~J. Wallman.
\newblock ``Bounding experimental quantum error rates relative to
  fault-tolerant thresholds''~(2015).

\bibitem{Kueng2016}
R.~Kueng, D.~M. Long, A.~C. Doherty, and S.~T. Flammia.
\newblock ``Comparing experiments to the fault-tolerance threshold''.
\newblock \href{https://dx.doi.org/10.1103/PhysRevLett.117.170502}{Phys. Rev.
  Lett. {\bf 117}, 170502}~(2016).

\bibitem{Kueng2017}
R.~Kueng, H.~Rauhut, and U.~Terstiege.
\newblock ``Low rank matrix recovery from rank one measurements''.
\newblock \href{https://dx.doi.org/10.1016/j.acha.2015.07.007}{Appl. Comput.
  Harmon. Anal. {\bf 42}, 88--116}~(2017).

\bibitem{Huang2020}
H.~Y. Huang, R.~Kueng, and J.~Preskill.
\newblock ``Predicting many properties of a quantum system from very few
  measurements''.
\newblock Nat. Phys. {\bf 16}, 1050––1057~(2020).
\newblock
  url:~\href{https://doi.org/10.1038/s41567-020-0932-7}{doi.org/10.1038/s41567-020-0932-7}.

\bibitem{Franca2020}
D.~S. Fran\c{c}a, F.~G. S.~L. Brand\~{a}o, and R.~Kueng.
\newblock ``{Fast and Robust Quantum State Tomography from Few Basis
  Measurements}''.
\newblock In Min-Hsiu Hsieh, editor, 16th Conference on the Theory of Quantum
  Computation, Communication and Cryptography (TQC 2021).
\newblock \href{https://dx.doi.org/10.4230/LIPIcs.TQC.2021.7}{Volume 197 of
  Leibniz International Proceedings in Informatics (LIPIcs), pages 7:1--7:13}.
\newblock Dagstuhl, Germany~(2021). Schloss Dagstuhl -- Leibniz-Zentrum f{\"u}r
  Informatik.

\bibitem{HIP_repo}
J.~Kahn.
\newblock ``Hyperplane intersection projection''.
\newblock
  \url{https://github.com/Hannoskaj/Hyperplane_Intersection_Projection}~(2021).

\bibitem{futureHIP}
J.~Kahn, M~Gu{\c{t}}{\u{a}}, R.~Kueng, and T.~Surawy-Stepney.
\newblock ``Hyperplane intersection projection''~(To be written).

\bibitem{Slater}
M.~Slater.
\newblock ``Lagrange multipliers revisited''.
\newblock Cowles Commission Discussion Paper No. 403~(1950).
\newblock
  url:~\href{https://cowles.yale.edu/sites/default/files/files/pub/d00/d0080.pdf}{cowles.yale.edu/sites/default/files/files/pub/d00/d0080.pdf}.

\bibitem{Boyd}
S.~Boyd and L.~Vandenberghe.
\newblock ``Convex optimization''.
\newblock
  \href{https://dx.doi.org/https://doi.org/10.1017/CBO9780511804441}{Cambridge
  University Press}. ~(2009).

\bibitem{broyden1970convergence}
Charles~George Broyden.
\newblock ``The convergence of a class of double-rank minimization algorithms
  1. general considerations''.
\newblock \href{https://dx.doi.org/https://doi.org/10.1093/imamat/6.1.76}{IMA
  Journal of Applied Mathematics {\bf 6}, 76--90}~(1970).

\bibitem{shanno1970conditioning}
David~F Shanno.
\newblock ``Conditioning of quasi-newton methods for function minimization''.
\newblock
  \href{https://dx.doi.org/https://doi.org/10.1090/S0025-5718-1970-0274029-X}{Mathematics
  of computation {\bf 24}, 647--656}~(1970).

\bibitem{fletcher1970new}
Roger Fletcher.
\newblock ``A new approach to variable metric algorithms''.
\newblock \href{https://dx.doi.org/https://doi.org/10.1093/comjnl/13.3.317}{The
  computer journal {\bf 13}, 317--322}~(1970).

\bibitem{goldfarb1970family}
Donald Goldfarb.
\newblock ``A family of variable-metric methods derived by variational means''.
\newblock \href{https://dx.doi.org/https://doi.org/10.2307/2004873}{Mathematics
  of computation {\bf 24}, 23--26}~(1970).

\bibitem{nesterov1983method}
Yurii Nesterov.
\newblock ``A method for unconstrained convex minimization problem with the
  rate of convergence o (1/k\^{} 2)''.
\newblock In Doklady an ussr.
\newblock Volume 269, pages 543--547.
\newblock ~(1983).
\newblock
  url:~\href{https://cir.nii.ac.jp/crid/1570572699326076416}{cir.nii.ac.jp/crid/1570572699326076416}.

\bibitem{doi:10.1137/19M1305045}
Alp Yurtsever, Joel~A. Tropp, Olivier Fercoq, Madeleine Udell, and Volkan
  Cevher.
\newblock ``Scalable semidefinite programming''.
\newblock \href{https://dx.doi.org/10.1137/19M1305045}{SIAM Journal on
  Mathematics of Data Science {\bf 3}, 171--200}~(2021).

\bibitem{Tol_colours}
P.~Tol.
\newblock ``Colour schemes''~(2018).

\bibitem{nielsen_chuang_2010}
M.~A. Nielsen and I.~L. Chuang.
\newblock ``Quantum computation and quantum information: 10th anniversary
  edition''.
\newblock \href{https://dx.doi.org/10.1017/CBO9780511976667}{Cambridge
  University Press}. ~(2010).

\bibitem{mat_b_ineq}
J.~A. Tropp.
\newblock ``User-friendly tail bounds for sums of random matrices''.
\newblock \href{https://dx.doi.org/10.1007/s10208-011-9099-z}{Foundations of
  Computational Mathematics {\bf 12}, 389--434}~(2012).

\bibitem{jensen1906}
J.~L. W.~V. Jensen.
\newblock ``Sur les fonctions convexes et les inégalités entre les valeurs
  moyennes''.
\newblock \href{https://dx.doi.org/10.1007/BF02418571}{Acta Math. {\bf 30},
  175--193}~(1906).

\end{thebibliography}

\cleardoublepage

\part*{Appendix}
\section{Form of the LS estimators in Section \ref{sec:experimental_procedures}}

\subsection{Proof of Proposition \ref{thm:ls_sc2}}\label{subsec:proof_of_ls2}
We start with a restatement of the form of the least-squares estimator for the Choi matrix, found using measurement scenario 2.
\begin{repproposition}{thm:ls_sc2}
For direct QPT with single-qubit Pauli inputs and single-qubit Pauli measurements,
the least-squares estimator for the Choi matrix $\Phi$ of a  $k$-qubit quantum channel takes the following form:
\begin{align}\label{app:ls_scenario_2}
\hat{\Phi}_{LS}=&\frac{1}{3^{2k}d}\sum_{\va\vb\vp\vq}f^{\va\vb}_{\vq\vp}\overset{k}{\underset{i=1}{\bigotimes}} M\:^{b_i}_{p_i} \,\overset{k}{\underset{j=1}{\bigotimes}}M\:^{a_j}_{q_j}, \quad \text{where} \\
& M\:^{b_i}_{p_i}= (3\ket{p_i,b_i}\!\bra{p_i,b_i}-\mathbb{1}_2), \\
& M\:^{a_j}_{q_j}=(3\ket{q_j,a_j}\!\bra{q_j,a_j}-\mathbb{1}_2).
\end{align}
\end{repproposition}

\begin{proof}
The proof of this result comes from relating the frequency measurements made in scenarios \ref{sub:pauli-ancilla} and \ref{sub:pauli-direct}. It relies on the following identity relating the output state $\C(\rho)$ of a channel $\C$ to the Choi matrix $\Phi= \C(|\Omega\rangle\langle \Omega |)$
\begin{equation}\label{sc_2_key_eq}
\C(\rho) = d\times\mbox{Tr}_{A}(\Phi\: (\mathbb{1}\otimes\rho^{\top}) ).
\end{equation}
Using this result, we have that the probability of observing the result $\vp\in\{0,1\}^{k}$ when measuring with setting $\textbf{b}$, having passed in the pure state $P^{\va\:\top}_{\vq}$ 
is:
\begin{equation}\label{probs2}
p^{\va\:\vb}_{\vq\:\vp}=d\times \mbox{Tr}(\Phi\:P^{\textbf{b}}_{\vp}\otimes P^{\va}_{\vq}).
\end{equation}
Note that up to the factor $d$, these coincide with the  probabilities \eqref{sc1_probs} from scenario \ref{sub:pauli-ancilla}, 
with the identifications $(\va\:\vb)= {\bf o}$ and $(\vq\:\vp)= {\bf s}$. The difference is due to the fact that while in scenario \ref{sub:pauli-ancilla} both parts $(\va\:\vb)$ of the outcome ${\bf o}$ are random, in scenario \ref{sub:pauli-direct} a single dataset is generated from random outcomes $\vb$ for each \emph{deterministically} chosen state with labels $\va$ and $\vq$.

Since the maps from the Choi matrix to probabilities are the same in the two scenarios, up to the factor $d$, the expression \eqref{ls_estimator_sc_1} of the LS estimator in scenario \ref{sub:pauli-ancilla} can be transferred to scenario  \ref{sub:pauli-direct}  by replacing the frequencies 
$f\:^{\bf s}_{\bf o}$ with $f\:^{\va\vb}_{\vq\vp}/d$ which gives us \eqref{app:ls_scenario_2}.

\end{proof}

%
%
%

%
%
%

In the rest of this section we give a longer, but perhaps more transparent proof of the equality
$$
p_{\bf o}^{\bf s} =\frac{1}{d} p^{\va\:\vb}_{\vq\:\vp}
$$
which may be of interest in its own right. We make the following comparisons between scenarios \ref{sub:pauli-ancilla} and \ref{sub:pauli-direct}.\par
In scenario \ref{sub:pauli-direct}, we have that the probability of observing the result $\vp\in\{\pm\}^{k}$ when measuring with setting $\textbf{b}$, having passed in the projector state $P^{\textbf{a}}_{\vq}$ is:
\begin{equation}
p^{\va\:\vb}_{\vq\:\vp}=\mbox{Tr}(\C( {P^{\va}_{\vq} }^\top)P^{\vb}_{\vp}).
\end{equation}
If we choose to prepare input state with indices $(\va,\vq)$ and measure with each setting $\vb$ an equal number of times, the construction of the estimator \eqref{app:ls_scenario_2} involves sampling equally from the $3^{2k}\times2^{k}$ corresponding distributions.

Alternatively, we could assign a probability of $\frac{1}{d}$ to choosing index $\vq$ while the setting indices $(\va,\vb)$ are deterministic as above; in this case the probability for recording joint outcome $(\vq,\vp)$ for an input/measurement setting of $(\va,\vb)$ is:
\begin{equation}\label{sc_1_prob}
\tilde{p}^{\va\:\vb}_{\vq\:\vp}=\frac{1}{d}\mbox{Tr}(\C({P^{\va}_{\vq} }^\top)P^{\vb}_{\vp}).
\end{equation}
In this case we are sampling from the $3^{2k}$ distributions corresponding to each setting $(\va,\vb)$.\\\par 
Consider now scenario \ref{sub:pauli-ancilla}. Observing the joint outcome $\vo=(\vq,\vp)$ when measuring the Choi matrix with setting $\vs=(\va,\vb)$ is equivalent to first measuring the ancilla qubits with setting $\va$, recording the outcome $\vq$, and then measuring the conditional system state
 with setting $\vb$ and recording outcome $\vp$. 
The probability of observing outcome $\vq$ when measuring with setting $\va$ is always $\frac{1}{d}$, since the marginal state of the ancilla is always the maximally mixed state. 
\par
The deferred measurement principle in quantum information theory \cite{nielsen_chuang_2010} implies that the this set up is equivalent to making a measurement on the marginal ancilla state \textit{before} acting with $\C\otimes\Id$ on the joint system. Defining $\Omega_{\vq}^{\:\:\va}$ as the marginal system state conditional on the measurement on the ancilla side of the maximally entangled state $\Omega$, we will now show that the probabilities in scenario 1 \eqref{sc1_probs} can be expressed as:
\begin{equation}\label{prob_1_marginal}
p^{\vs}_{\vo} =\tilde{p}^{\va\:\vb}_{\vq\:\vp}=\frac{1}{d}\mbox{Tr}(\C(\Omega_{\vq}^{\:\:\va})\:P^{\textbf{b}}_{\vp}).
\end{equation}\\\par
Let $\{\ket{\vq}=\overset{k/2}{\underset{i=1}{\bigotimes}}\ket{q_i}\}$ be the standard basis in $\mathbb{C}^d$.
For each different measurement setting $\va$, the corresponding Pauli basis $\{\ket{{\vq},{\va}}\}$, is related to $\{ \ket{\vq}\}$ by a unitary transformation $U^{\va}$ 
\begin{equation}
U^{\va}\ket{\vq}=\ket{{\vq},{\va}}
\end{equation}
and the associated projection operators transform as
\begin{equation}\label{introducing_p_tilde}
P^{\va}_{\vq}=U^{\va}P_{\vq}(U^{\va})^{*}.
\end{equation}
The maximally entangled input state $\ket{\omega}$ in the standard basis is:
\begin{equation}
\ket{\omega}=\frac{1}{\sqrt{d}}\sum_{\vp}\ket{\vp}\otimes\ket{\vp}.
\end{equation} 
Therefore, after measuring the ancilla qubits with setting $\va$ on the joint state $\ket{\omega}$ and obtaining outcome $\vq$, the projected state is
\begin{equation}
\begin{aligned}
\mathbb{1}_d \otimes P^{\va}_{\vq}\ket{\omega}&=(\mathbb{1}_d\otimes(U^{\va}P_{\vq}(U^{\va})^{*})\ket{\omega}\\
&=(\mathbb{1}_d\otimes(U^{\va}P_{\vq}(U^{\va})^{*})\bar{U}^{\va}\otimes U^\va\ket{\omega}\\
&=(\bar{U}^{\va}\otimes U^{\va}P_{\vq})\ket{\omega} \\ 
&=(\bar{U}^{\va}\otimes U^{\va})\frac{1}{\sqrt{d}}\ket{\vq}\otimes\ket{\vq}.
\end{aligned}
\end{equation}
Therefore, the conditional state of the system after making the above measurement on the ancilla qubits is $\bar{U}^{\va}\ket{\vq}$. We conclude, then, that \begin{equation}\Omega^{\:\:\va}_{\vq}=\bar{U}^{\va}\ket{\vq}\!\bra{\vq}(\bar{U}^{\va})^{*}.\end{equation} 
But $\ket{\vq}\!\bra{\vq}=P_{\vq}=P_{\vq}^{\top}$, so:
\begin{equation}
\Omega^{\:\:\va}_{\vq}=P^{\va\:\top}_{\vq}
\end{equation}
Therefore the probabilities obtained making measurements in scenario 1 are:
\begin{equation}\label{prob_1_marginal2}
p^{\vs \:=\: \va\:\vb}_{\vo \:=\: \vq\:\vp}=\frac{1}{d}\mbox{Tr}(\C(P^{\va\:\top}_{\vq})\:P^{\textbf{b}}_{\vp}).
\end{equation}
which coincides with equation \eqref{prob_1_marginal} corresponding to the alternative scenario 2, where the input indices $\vq$ were chosen randomly with uniform probability.
\qed

\subsection{Proof of Proposition \ref{th.LS.scenario4.general}}
\label{proof.th.LS.scenario4.general}

\begin{repproposition}{th.LS.scenario4.general}
For direct QPT with (transposed) MUB inputs $(|w_i \rangle \! \langle w_i|)^\top$ and MUB measurements $\tfrac{d}{m}|v_l \rangle \! \langle v_l|$, the least-squares estimator for the Choi matrix $\Phi$ of a $k$-qubit quantum channel takes the following form: 
\begin{eqnarray}
\hat{\Phi}_{LS}  &=&
\frac{d+1}{d} \sum_{l,k=1}^m f_l^k   \ket{v_l}\!\bra{v_l} \otimes   \ket{w_k}\!\bra{w_k} \nonumber \\
&& - \frac{1}{d}\sum_{l,k=1}^m f_l^k  (  \ket{v_l}\!\bra{v_l} \otimes \mathbb{1}_d + \mathbb{1}_d \otimes  \ket{w_k}\!\bra{w_k} )
\nonumber \\
&&+ \mathbb{1}_d  \otimes  \mathbb{1}_d.
\label{app:eq.LS.scenario4.general}
\end{eqnarray}
\end{repproposition}

\begin{proof}
We apply the relation 
$$
\C(\rho) = d\times\mbox{Tr}_{A}(\Phi\: (\mathbb{1}\otimes\rho^{\top}) ).
$$
to express the probability distribution for a give input state $\ket{w_k}\!\bra{w_k}^\top$ and measurement with POVM elements 
$M_l=\frac{d}{m}\ket{v_l}\!\bra{v_l}$ as

\begin{equation}
p^{k}_{l}=\mbox{Tr}(\C(\ket{w_k}\!\bra{w_k}^\top)\:M_l) = \frac{d^2}{m} {\rm Tr} (\Phi\: ( P_l \otimes Q_k)) 
\end{equation}
where $P_l = \ket{v_l}\!\bra{v_l}, Q_k = \ket{w_k}\!\bra{w_k}$.
The linear map from Choi matrices to the collection of probability distributions is then
\begin{eqnarray}
\mathcal{A} : M(\mathbb{C}^d) &\to & \mathbb{C}^{m}\otimes  \mathbb{C}^m  \cong \mathbb{C}^{m^2}
\\
\mathcal{A} : \Phi &\mapsto& p:=  \{ p_l^k  : l,k=1,\dots ,m\}
\end{eqnarray}
It's adjoint is 
$$
\mathcal{A}^\dagger : q \mapsto \frac{d^2}{m}  \sum_{l,k=1}^m q_l^k \,  P_l \otimes Q_k.
$$

The general form of the LS estimator, assuming $\mathcal{A}^\dagger \mathcal{A}$ is invertible, is given by 
\begin{equation}\label{eq:general_ls_form}
\hat{\Phi} = (\mathcal{A}^\dagger \mathcal{A})^{-1}  \mathcal{A}^\dagger f
\end{equation}
where $f:= \{ p_l^k  : l,k=1,\dots ,m\}$ is the vector of empirical frequencies. To find a more explicit expression and prove the validity of equation \eqref{eq:general_ls_form}, we analyse 
$\mathcal{A}^\dagger \mathcal{A}$ in more detail. For this we use the relations characteristic for 2-designs
\begin{eqnarray*}
&&
\sum_{k=1}^m {\rm Tr} ( X P_l )P_l = X+ {\rm Tr} (X) \mathbb{1}_d\\
&&
\sum_{k=1}^m {\rm Tr} ( X Q_k )Q_k = X+ {\rm Tr} (X) \mathbb{1}_d
\end{eqnarray*}
We now show that $\mathcal{A}^\dagger \mathcal{A}$ has 3 eigevalues with eigenspaces consisting of linear spans of the following 4 types of operators:\\

1. Let $\Lambda_1= X\otimes Y$ with ${\rm Tr}(X)={\rm Tr} (Y) =0$. In this case
\begin{eqnarray*}
\mathcal{A}^\dagger \mathcal{A}(\Lambda_1) 
&=& 
\frac{d^4}{m^2} \sum_{l,k=1}^m  {\rm Tr} (Z ( P_l \otimes Q_k))  P_l \otimes Q_k\\
&=&
\frac{d^4}{m^2} \left(\sum_{l=1}^m   {\rm Tr} (X  P_l ) P_l\right)\left(\sum_{k=1}^m   {\rm Tr} (Y  Q_k) Q_k\right)\\
&=&
\frac{d^4}{m^2} X\otimes Y = \frac{d^2}{(d+1)^2} \Lambda_1
\end{eqnarray*}\\

2. Let $\Lambda_2= X\otimes \mathbb{1}_d$  with ${\rm Tr}(X)=0$. In this case
\begin{eqnarray*}
\mathcal{A}^\dagger \mathcal{A}(\Lambda_2) 
&=& 
\frac{d^4}{m^2} \sum_{l,k=1}^m  {\rm Tr} (\Lambda_2 ( P_l \otimes Q_k))  P_l \otimes Q_k\\
&=& \left(\sum_{l=1}^m   {\rm Tr} (X  P_l ) P_l\right) \left(\sum_{k=1}^m   {\rm Tr} (Q_k) Q_k \right)\\
&=& \frac{d^4 (d+1)}{m^2}  X \otimes \mathbb{1}_d =
\frac{d^2}{d+1} \Lambda_2
\end{eqnarray*}
where in the last step we used that $\sum_k Q_k = (d+1) \mathbb{1}_d$.\\

3. A similar equality holds for $\Lambda_3=  \mathbb{1}_d \otimes Y$ with ${\rm Tr}(Y) =0$.\\

4. Let $\Lambda_4= {\mathbb{1}_d\otimes \mathbb{1}_d}$. Here, we have
\begin{eqnarray*}
\mathcal{A}^\dagger \mathcal{A}(\Lambda_4) 
&=& 
\frac{d^4}{m^2} \sum_{l,k=1}^m  {\rm Tr} (Z ( P_l \otimes Q_k))  P_l \otimes Q_k\\
&=& \frac{d^4}{m^2} \left(\sum_{l=1}^m  P_l\right) \left(\sum_{k=1}^m  Q_k \right)\\
&=& \frac{d^4(d+1)^2 }{m^2} \mathbb{1}_d\otimes \mathbb{1}_d = d^2 \Lambda_4
\end{eqnarray*}

Writing $P_l = \tilde{P}_l + \mathbb{1}_d/d$ with ${\rm Tr} ( \tilde{P}_l) =0$, and $Q_k = \tilde{Q}_k + \mathbb{1}_d/d$ with 
${\rm Tr} ( \tilde{Q}_k) =0$, we now have
\begin{eqnarray*}
\left(\mathcal{A}^\dagger \mathcal{A}\right)^{-1} \mathcal{A}^\dagger f &&=
\left(\mathcal{A}^\dagger \mathcal{A}\right)^{-1} \frac{d^2}{m}\sum_{l,k=1}^m f_l^k P_l\otimes Q_k\\
&&= 
\left(\mathcal{A}^\dagger \mathcal{A}\right)^{-1} \frac{d^2}{m}\sum_{l,k=1}^m f_l^k ( \tilde{P}_l + \frac{\mathbb{1}_d}{d})  \otimes( \tilde{Q}_k + \frac{\mathbb{1}_d}{d}).
\end{eqnarray*}
By expanding the brackets and using the fact that each term is an eigenvector of $\mathcal{A}^\dagger \mathcal{A}$, taking the form of one of $Z_{(1,2,3,4)}$ above, we get
\begin{eqnarray*}
\left(\mathcal{A}^\dagger \mathcal{A}\right)^{-1} \mathcal{A}^\dagger f &=&
\frac{d^2}{m} \frac{(d+1)^2}{d^2}\sum_{l,k=1}^m f_l^k \tilde{P}_l\otimes \tilde{Q}_k \\
&&+ \frac{d^2}{m} \frac{d+1}{d^2} \sum_{l,k=1}^m f_l^k \tilde{P}_l\otimes \frac{\mathbb{1}_d}{d}\\
&&+ \frac{d^2}{m} \frac{d+1}{d^2} \sum_{l,k=1}^m f_l^k \frac{\mathbb{1}_d}{d}\otimes \tilde{Q}_k\\
&&+  \frac{d^2}{m} \frac{1}{d^2} \sum_{l,k=1}^m f_l^k \frac{\mathbb{1}_d}{d} \otimes  \frac{\mathbb{1}_d}{d}
\end{eqnarray*}
By expressing $\tilde{P}_l = P_l - \mathbb{1}_d/d$ and $\tilde{Q}_k = Q_k - \mathbb{1}_d/d$ we get
\begin{eqnarray*}
\left(\mathcal{A}^\dagger \mathcal{A}\right)^{-1} \mathcal{A}^\dagger f &&= \frac{d+1}{d} \sum_{l,k=1}^m f_l^k  (P_l - \frac{\mathbb{1}_d}{d}) \otimes ( Q_k - \frac{\mathbb{1}_d}{d})\\
&&+\frac{1}{d} \sum_{l,k=1}^m f_l^k 
\left[ 
(P_l - \frac{\mathbb{1}_d}{d}) \otimes \frac{\mathbb{1}_d}{d }  + 
\frac{\mathbb{1}_d}{d } \otimes(Q_k - \frac{\mathbb{1}_d}{d})    \right]\\
&&+  \frac{\mathbb{1}_d}{d} \otimes \frac{\mathbb{1}_d }{d}
\end{eqnarray*}
Finally, by rearranging terms we get
\begin{eqnarray*}
\left(\mathcal{A}^\dagger \mathcal{A}\right)^{-1} \mathcal{A}^\dagger f &&=
\frac{d+1}{d} \sum_{l,k=1}^m f_l^k  P_l  \otimes  Q_k \\
&& - \frac{1}{d}\sum_{l,k=1}^m f_l^k  ( P_l \otimes \mathbb{1}_d + \mathbb{1}_d \otimes Q_k)\\
&&+ \mathbb{1}_d  \otimes  \mathbb{1}_d
\end{eqnarray*}
\end{proof}

\section{Concentration bounds for the LS estimators}\label{app:ls_conc_bounds}

In order to prove Theorem \ref{thm:projected_bounds}, we first find error bounds for the least-squares estimators for each of the scenarios \ref{sub:pauli-ancilla}-
\ref{sub:MUB-direct} in Frobenius norm distance, in sections \ref{sec:ls_bound_sc1}-\ref{sec:ls_bound_sc4}. For this we will be applying a matrix generalisation of the classical scalar Bernstein inequality, see Ref.~\cite{mat_b_ineq}.
\begin{theorem}\label{Th.Bernstein}
Consider a sequence of N \textit{independent}, Hermitian, random matrices $\{A_i\}_{i=1}^N\:\in\mathbb{C}^{n\times n}$ satisfying
\textup{\begin{equation}
\mathbb{E}[A_k]=0\:\:\mbox{and}\:\:\|A_k\|_\infty\leq R.
\end{equation}}
Then, for $t\geq0$:
\textup{\begin{equation}\label{matrix_bernstein_inequality}
\mbox{Pr}[\|\sum_{i=1}^NA_i\|_\infty\geq t] \leq 
N \exp\left(-\frac{t^2/3}{\sigma^2+Rt/3} \right),
\end{equation}}
where \textup{$\sigma^2=\|\sum_{i=1}^N\mathbb{E}[(A_i)^2]\|_\infty$} captures the norm of the total variance.
\end{theorem}

\subsection{Scenario \ref{sub:pauli-ancilla}}\label{sec:ls_bound_sc1}
\begin{proposition}\label{prop.concentration.sc.1}
The operator norm distance of the least-squares estimator from the true Choi matrix for scenario \ref{sub:pauli-ancilla} satisfies the following bound:
\textup{\begin{equation}\label{pls_bound_0}
\mbox{Pr}[\|\hat{\Phi}_{LS}-\Phi\|_{\infty}\geq\tau]\:\leq\:d^2\mbox{exp}\left(-\frac{3N\tau^2}{8\times3^{2k}} \right)
\end{equation}}
for a $2k$-qubit joint system-ancilla system, where $d=2^k$, $N$ is the total number of measurements, $\hat{\Phi}_{LS}$ is the least-squares estimator for the true state $\Phi$, and $\tau\in[0,1]$.
\end{proposition}
\begin{proof}
We construct this bound by borrowing the result for the \lq Pauli basis measurements\rq\:case given in \cite{Gu__2020} and adapting it for a larger system. The proof is an application of the matrix Bernstein Inequality and we sketch it here, though the reader is directed to \cite{Gu__2020} for details.\\\par 
We can write the LS estimator \eqref{ls_estimator_sc_1} in the following way:
\begin{equation}\label{ls_sc_1_random_matrices}
\hat{\Phi}_{LS}=\sum_{\vs}\frac{1}{N}\sum_i^{N/3^{2k}}X_i^\vs
\end{equation}
where, for every $\vs\in\{x,y,z\}^{2k}$, the $X_i^\vs$ are independent instances of the random matrices:
\begin{equation}
X^\vs :{\bf o}\mapsto \overset{2k}{\underset{i=1}{\otimes}}(3\ket{o_i,s_i}\!\bra{o_i,s_i}-\mathbb{1}_{d^2})
\end{equation}
with probability distribution over ${\bf o}\in \{0,1\}^{2k}$
\begin{equation}
p_\vo^\vs=\mbox{Tr}(\Phi\:P^\vs_\vo).
\end{equation}
We can then write 
\begin{equation}\label{ls_phi_infty_mbe_1}
\|\hat{\Phi}_{LS}-\Phi\|_\infty=\left\|\sum_{\vs}\frac{1}{N}\sum_i^{N/3^{2k}}(X_i^\vs-\mathbb{E}[X^\vs])\right\|_\infty.
\end{equation}
If we define the matrices:
\begin{equation}
A^{\bf s}_i=\frac{1}{N}(X_i^\vs-\mathbb{E}[X_i^\vs]),
\end{equation}
equation \eqref{ls_phi_infty_mbe_1} becomes
\begin{equation}
\|\hat{\Phi}_{LS}-\Phi\|_\infty=\|\sum_{\vs}\sum_i^{N/3^{2k}}A^{\bf s}_i\|_\infty.
\end{equation}
Thus we can apply the matrix Bernstein inequality to bound $\|\hat{\Phi}_{LS}-\Phi\|_\infty$.
We have that:
\begin{equation}\label{var_sc_1}
\sigma^2=\left\|\sum_{\vs}\frac{1}{N^2}\sum_i^{N/3^{2k}} \mathbb{E} [(X_i^\vs-\mathbb{E}[X^\vs])^2 ]\right\|_\infty
\end{equation}
and
\begin{equation}\label{R_sc_1}
R=\mbox{max}\{\frac{1}{N}\|X_i^\vs-\mathbb{E}[X_i^\vs]\|_\infty\}.
\end{equation}
Using the following inequality:
\begin{equation}
\mathbb{E}[(X_i^\vs-\mathbb{E}[X_i^\vs])^2]\leq\mathbb{E}[(X^\vs)^2]
\end{equation}
and Jensen's inequality \cite{jensen1906}, we can solve \eqref{var_sc_1} and \eqref{R_sc_1} to find:
\begin{equation}
\sigma^2=\frac{3^k}{N}\:\:\:\:\mbox{and}\:\:\:\:R=\frac{2^{k+1}}{N}.
\end{equation}
Applying the matrix Bernstein inequality then gives us equation \eqref{pls_bound_0} , provided that $\tau$ is not too large.
\end{proof}\par

We now convert this bound into one containing the Frobenius norm. 
\begin{proposition}\label{prop:ls_2rob_bound_sc1}
The following bound holds
\textup{\begin{equation}\label{pls_bound_2}
\mbox{Pr}[\|\hat{\Phi}_{LS}-\Phi\|_2^2\geq\delta^2]\:\leq\:d^2\mbox{exp}(-\frac{3N\delta^2}{8\times d^2\times3^{2k}})
\end{equation}}
for $\delta^2\in[0,1]$.
\end{proposition}
\begin{proof}
Taking \eqref{pls_bound_1} as our starting point, we can make use of the following inequality involving the Frobenius norm and operator norm of a Hermitian matrix:\\
Let $A\in\mbox{Herm}_{n\times n}$ be an arbitrary Hermitian matrix, then
\begin{equation}\label{frob_op_rltn}
\|A\|_2^{2}\:\leq\:d^2\|A\|_{\infty}^{2}.
\end{equation}
$\hat{\Phi}_{LS}-\Phi$ is clearly Hermitian since it is constructed as a sum over a set of random Hermitian matrices.\par
If $\|\hat{\Phi}_{LS}-\Phi\|_{\infty}\leq\tau$ with some probability, it must be that $\|\hat{\Phi}_{LS}-\Phi\|_2^{2}\leq d^2\tau^2$ with at least that probability, by equation \eqref{frob_op_rltn}. Therefore, setting $\delta^2=d^2\tau^2$ we arrive at \eqref{pls_bound_2}. This is valid for $\delta^2\in[0,d^2]$, so is too in the region $[0,1]$ that we are interested in.\\
\end{proof}

\subsection{Scenario \ref{sub:pauli-direct}}\label{variance_scenario_2}

\begin{proposition}
The LS estimator in scenario \ref{sub:pauli-direct} satisfies the concentration bound
\textup{\begin{equation}\label{pls_bound_1}
\mbox{Pr}[\|\hat{\Phi}_{LS}-\Phi\|_{\infty}\geq\tau]\:\leq\:d^2\mbox{exp}(-\frac{3N\tau^2}{8\times3^{2k}})
\end{equation}}
where $N$ is the total number of measurements, $\hat{\Phi}_{LS}$ is the least-squares estimator for the true state $\Phi$, and $\tau\in[0,1]$.

\end{proposition}

\begin{proof}
Using Proposition \ref{thm:ls_sc2} we can express the least squares estimator in scenario \ref{sub:pauli-direct} as a sum of random matrices, similarly to scenario \ref{sub:pauli-ancilla}.

\begin{equation*}
\hat{\Phi}_{LS}=
\sum_{{\bf a},{\bf b},{\bf q}}\frac{1}{N}\sum_i^{N/s }X^{{\bf a},{\bf b}}_{{\bf q},i}
\end{equation*}
where $s= 3^{2k}\cdot 2^{k}$, and   
${\bf a} , {\bf b}\in\{x,y,z\}^{k}$, ${\bf q}\in \{0,1\}^{k}$. Here, $X^{{\bf a},{\bf b}}_{{\bf q},i}$ are independent instances of the random matrices:
\begin{equation*}
X^{{\bf a},{\bf b}}_{{\bf q}} :{\bf p}\mapsto \overset{k}{\underset{i=1}{\bigotimes}} M\:^{b_i}_{p_i} \,\overset{k}{\underset{j=1}{\bigotimes}}M\:^{a_j}_{q_j}
\end{equation*}
where 
$$
M\:^{b_i}_{p_i}=(3\ket{p_i, b_i}\!\bra{p_i, b_i}-\mathbb{1}_2), \quad 
M\:^{a_j}_{q_j}=(3\ket{q_j,a_j}\!\bra{q_j,a_j}-\mathbb{1}_2).
$$
with probability distribution over ${\bf p}\in \{0,1\}^k$
\begin{equation*}
p^{\va\:\vb}_{\vq\:\vp}=\mbox{Tr}(\C( {P^{\va}_{\vq} }^\top)P^{\vb}_{\vp}).
\end{equation*}
Then
$$
\hat{\Phi}_{LS} -\Phi = \sum_{{\bf a},{\bf b},{\bf q}}\sum_i^{N/s }A^{{\bf a},{\bf b}}_{{\bf q},i}
$$
where $A^{{\bf a},{\bf b}}_{{\bf q},i} = \frac{1}{N} (X^{{\bf a},{\bf b}}_{{\bf q},i} -\mathbb{E} (X^{{\bf a},{\bf b}}_{{\bf q},i} ))$ and we can apply the concentration bound from Theorem \ref{Th.Bernstein} similarly to Proposition \ref{prop.concentration.sc.1}. 
Since $X^{{\bf a},{\bf b}}_{{\bf q}}$ takes values in the same set as the random matrix $X^{\bf s}$ in Proposition \ref{prop.concentration.sc.1} (scenario \ref{sub:pauli-ancilla}), we obtain
$$
R =\max_{{\bf a},{\bf b},{\bf q}}  \frac{1}{N} \| X^{{\bf a},{\bf b}}_{{\bf q}}- \mathbb{E} (X^{{\bf a},{\bf b}}_{{\bf q}} ) \|_\infty = \frac{2^{k+1}}{N}.
$$
We now consider the variance
$$
\sigma^2=
\left\|\sum_{{\bf a},{\bf b},{\bf q}} \frac{1}{N^2} \sum_i^{N/s} 
\mathbb{E} [
(X^{{\bf a},{\bf b}}_{{\bf q},i}-\mathbb{E}[X^{{\bf a},{\bf b}}_{\bf q}])^2 ]\right\|_\infty 
$$
Using $\mathbb{E} [(X = \mathbb{E} (X))^2 ] \leq \mathbb{E}[ X^2]$ we get
$$
\sigma^2\leq \left\|\sum_{{\bf a},{\bf b},{\bf q}} \frac{1}{sN} 
\mathbb{E} [
(X^{{\bf a},{\bf b}}_{{\bf q}})^2 ]\right\|_\infty 
$$
Now using the relation between scenario \ref{sub:pauli-ancilla} and \ref{sub:pauli-direct}  probabilities $p^{{\bf a} , {\bf b}}_{{\bf q} , {\bf p}} =d p^{\bf s}_{\bf o}$ where ${\bf s} = ({\bf a}, {\bf b})$ and ${\bf o} = ({\bf q} , {\bf p})$, we have
\begin{eqnarray*}
\sum_{{\bf a},{\bf b},{\bf q}} \frac{1}{sN} 
\mathbb{E} [
(X^{{\bf a},{\bf b}}_{{\bf q}})^2 ]& =& 
\frac{1}{sN}\sum_{{\bf a},{\bf b},{\bf q}}  \sum_{\bf p} p^{{\bf a} , {\bf b}}_{{\bf q} , {\bf p}} 
(X^{{\bf a},{\bf b}}_{{\bf q}} (p))^2 \\
&=&d \sum_{{\bf o, \bf s}} p^{\bf s}_{\bf o} (X^{\bf s} ({\bf o}))^2 \\
&=& d \sum_{\bf s}\mathbb{E}[ (X^{\bf s})^2].
\end{eqnarray*}
Taking into account that $s= d 3^{2k}$ we arrive at the same upper bound expression as found in scenario \ref{sub:pauli-ancilla}
\begin{eqnarray*}
\sigma^2 &&\leq \| \frac{d}{sN}  \sum_{\bf s} \mathbb{E} [ (X^{\bf s})^2] \| _\infty \\
&&= \frac{1}{3^{2k}N}  \|   \sum_{\bf s} \mathbb{E} [ (X^{\bf s})^2] \| _\infty\\
&&= \frac{3^{2k}}{N}.
\end{eqnarray*}

\end{proof}
\begin{proposition}\label{prop:ls_2rob_bound_sc2}
The following bound holds
\textup{\begin{equation}\label{pls_bound_22}
\mbox{Pr}[\|\hat{\Phi}_{LS}-\Phi\|_2^2\geq\delta^2]\:\leq\:d^2\mbox{exp}(-\frac{3N\delta^2}{8\times d^2\times3^{2k}})
\end{equation}}
for $\delta^2\in[0,1]$.
\end{proposition}


\subsection{Scenario \ref{sub:MUB+ancilla}}\label{sec:LS_sc3}

\begin{proposition}\label{prop:ls_op_norm_bound_sc_real_3}
The operator norm error of the least-squares estimator in scenario \ref{sub:MUB+ancilla} satisfies the following bound:
\textup{\begin{equation}\label{ls_op_norm_bound_sc_real_3}
\mbox{Pr}[\|\hat{\Phi}_{LS}-\Phi\|_{\infty}\geq\tau]\:\leq\:d^2\mbox{exp}(-\frac{3N\tau^2}{16}\frac{1}{d^2}).
\end{equation}}
\end{proposition}
This is a direct application of the result for 2-design POVMs in \cite{Gu__2020} to a system of dimension of $d^2$. It serves no purpose to repeat the proof - another application of the matrix-Bernstein inequality - here, and the reader is directed to \cite{Gu__2020} for an exposition.
\begin{proposition}
\textup{\begin{equation}\label{ls_2rob_norm_bound_sc_real_3}
\mbox{Pr}[\|\hat{\Phi}_{LS}-\Phi\|_2^2\geq\delta^2]\:\leq\:d^2\mbox{exp}(-\frac{3N\delta^2}{16}\frac{1}{d^4}).
\end{equation}}
\end{proposition}
\begin{proof}
As with scenario \ref{sub:pauli-ancilla}, we use the relation \eqref{frob_op_rltn} to derive the extra factor of $\frac{1}{d^2}$ in the exponential, and arrive at \eqref{ls_2rob_norm_bound_sc_real_3}.
\end{proof}

\subsection{Scenario \ref{sub:MUB-direct}}\label{sec:ls_bound_sc4}

\begin{proposition}\label{prop:ls_op_norm_bound_sc_3}
The operator norm error of the least-squares estimator in  scenario \ref{sub:MUB-direct} satisfies the following bound:
\textup{\begin{equation}\label{ls_op_norm_bound_sc_3}
\mbox{Pr}[\|\hat{\Phi}_{LS}-\Phi\|_{\infty}\geq\tau]\:\leq\:d^2\mbox{exp}
\left(-\frac{3N\tau^2}{32d^2}\right)
\end{equation}}
\end{proposition}

\begin{proof}
We return to another application of the concentration inequality \eqref{matrix_bernstein_inequality}.\\\par
Once again, define $P_l=\ket{v_l}\!\bra{v_l}$ and $Q_k=\ket{w_k}\!\bra{w_k}$. Define also $Y^k_l = \frac{d+1}{d}P_l\otimes Q_k$ and $Z^k_l = \frac{1}{d}(P_l\otimes\mathbb{1}_d + \mathbb{1}_d\otimes Q_k)$ for future use.\\\par
For each $k \in [m]$ we define the random matrices
\begin{equation}
X^k : l \mapsto \frac{d+1}{d}P_l\otimes Q_k - \frac{1}{d}(P_l\otimes\mathbb{1}_d + \mathbb{1}_d\otimes Q_k)
\end{equation}
with probability distribution over $l \in \{1, ..., m\}$
\begin{equation}
p_l^k=\mbox{Tr}(\C(Q_k^\top)M_l).
\end{equation}
We denote by $X^k_i$ the independent samples of $X^k$ corresponding to the measurement outcomes for input state $Q_k^\top$, with $i=1, ..., N/m$.

We then write the LS estimator 
\eqref{app:eq.LS.scenario4.general} in the following way:
\begin{equation}\label{ls_sc_3_random_matrices}
\Pls=\sum_{k=1}^{(d+1)d}\frac{m}{N}\sum_{i=1}^{N/m}X_i^k + \mathbb{1}_d\otimes\mathbb{1}_d,
\end{equation}
assuming again that the N measurements are distributed uniformly among the different input states.\par
Then, we can write the least-squares estimator operator-norm distance as
\begin{equation}
\begin{aligned}
\|\Pls-\Phi\|_\infty&=\left\|\sum_{k=1}^{(d+1)d}\frac{m}{N}\sum_{i=1}^{N/m}X_i^k-\sum_{k=1}^{(d+1)d}\mathbb{E}[X^k]\right\|_\infty\\
&=\left\|\sum_{k=1}^{(d+1)d}\sum_{i=1}^{N/m}A_i^k \right\|_\infty
\end{aligned}
\end{equation}
i.e. the operator norm of a sum over independent centred random matrices:
\begin{equation}
A_i^k=\frac{m}{N}(X_i^k-\mathbb{E}[X^k])\mbox{;}\:\:\:i\in[N],\:\:k\in[(d+1)d].
\end{equation}
The upper bound, $R$, for the operator norm of each contribution to the sum (see equation \eqref{matrix_bernstein_inequality}) can be found in the following way:
\begin{equation}
\begin{aligned}\label{R_sc_3}
\|A_i^k\|_\infty&\leq\frac{m}{N}(\|X_i^k\|_\infty+\|\mathbb{E}[X_i^k]\|_\infty)\\
&\leq \frac{m}{N}(\| X_i^k\|_\infty)(1+\mbox{max}\{p^k_l\})\\
&\leq \frac{2m}{N}(\| Y_i^k\|_\infty+\| Z_i^k\|_\infty)\\
&=\frac{2(d+1)(d+3)}{N}=:R.
\end{aligned}
\end{equation}
And the variance term:

\begin{equation}
\begin{aligned}
\sigma^2&=\left\|\sum_{k=1}^m\sum_{i=1}^{N/m}\mathbb{E}[(A_i^k)^2]\right\|_\infty\\
&=\left\|\sum_{k=1}^{m}\frac{m^2}{N^2}\sum_{i=1}^{N/m}\mathbb{E}[(X_i^k-\mathbb{E}[x_i^k])^2] \right\|_\infty\\
&=\frac{m}{N}\|\sum_{k=1}^{m}\mathbb{E}[(X^k-\mathbb{E}[X^k])^2]\|_\infty\\
&\leq\frac{m}{N}\|\sum_{k=1}^m\mathbb{E}[(X^k)^2]\|_\infty
\end{aligned}
\end{equation}


where the inequality is with respect to Loewner ordering. Now
\begin{equation}
\begin{aligned}
\sum_{k=1}^m\mathbb{E}[(X^k)^2]&=\sum_{k,l=1}^m p^k_l(Y_l^k - Z^k_l)^2\\
&=\frac{1}{d^2}\sum_{k,l=1}^m p^k_l(d(d-1)P_l\otimes Q_k - dX^k_l)\\
&=\frac{d-1}{d}\sum_{k,l=1}^m p^k_l P_l\otimes Q_k - \frac{1}{d}\sum_{k=1}^m\mathbb{E}[X^k].
\end{aligned}
\end{equation}
Given that 
\eqref{app:eq.LS.scenario4.general}
is unbiased, we have that $\sum_{k=1}^m\mathbb{E}[X^k] = \Phi-\mathbb{1}_d\otimes\mathbb{1}_d$.
Focussing on the first term:
\begin{equation}
\begin{aligned}
\sum_{k,l=1}^m p^k_l P_l\otimes Q_k &= \sum_{k,l=1}^m \mbox{tr}(\frac{d}{m}P_l\C(Q_k^\top)) P_l\otimes Q_k\\
&= \frac{d^2}{m}\sum_{k,l=1}^m \mbox{tr}(P_l\otimes Q_k \Phi) P_l\otimes Q_k\\
\end{aligned}
\end{equation}
by equation \eqref{sc_2_key_eq}.\par
We now use the fact that the $P_l$ and $Q_k$ form spherical 2-designs. The near-isotropy of such 2-designs ensures the following for arbitrary selfadjoint operators $A, B$:
\begin{equation}
\begin{aligned}
\frac{d^2}{m}\sum_{k,l=1}^m \mbox{tr}((P_l\otimes Q_k)&(A\otimes B)) P_l\otimes Q_k = d^3(d+1)(\frac{1}{m}\sum_{l=1}^mP_l\mbox{tr}(P_l A))(\frac{1}{m}\sum_{l=1}^mQ_k\mbox{tr}(Q_k B))\\
&= \frac{d}{d+1}(A+\mbox{tr}(A)\mathbb{1}_d)(B+\mbox{tr}(B)\mathbb{1}_d)\\
&= \frac{d}{d+1}(A\otimes B + A\otimes \mbox{tr}(B)\mathbb{1}_d + \mbox{tr}(A)\mathbb{1}_d\otimes B+ \mbox{tr}(A)\mbox{tr}(B)\mathbb{1}_d\otimes\mathbb{1}_d).
\end{aligned}
\end{equation}
This can be extended linearly to all selfadjoint $\Phi\in M(\mathbb{C}^d)^{\otimes 2}$ such that:
\begin{equation}
\begin{aligned}
\frac{d^2}{m}\sum_{k,l=1}^m \mbox{tr}(P_l\otimes Q_k \Phi) P_l\otimes Q_k =\frac{d}{d+1}(\Phi + \mbox{tr}_A(\Phi)\otimes\mathbb{1}_d+\mathbb{1}_d\otimes\mbox{tr}_S(\Phi) + \mbox{tr}(\Phi)\mathbb{1}_d\otimes\mathbb{1}_d).
\end{aligned}
\end{equation}
Thus:
\begin{equation}
\begin{aligned}
\sigma^2&\leq\frac{m}{N}\| \sum_{k=1}^m\mathbb{E}[(X^k)^2]\|_\infty \\
&\leq \left\|\frac{d-1}{d+1}(\Phi + \mbox{tr}_A(\Phi)\otimes\mathbb{1}_d+\mathbb{1}_d\otimes\mbox{tr}_S(\Phi) \right. \\
&\:\:\:\:\:\:\:\:\:\:\:\:\:\:\:\:\:\:\:\:\:\:\:\:\:\:\:\:\:\:\:\:\:\: + \mbox{tr}(\Phi)\mathbb{1}_d\otimes\mathbb{1}_d)) \\
&\:\:\:\:\:\:\:\:\:\:\:\:\:\:\:\:\:\:\:\:\:\:\:\:\:\:\:\:\:\:\:\:\:\: \left. -\frac{1}{d}(\Phi+\mathbb{1}_d\otimes\mathbb{1}_d))  \right\|_\infty\\
&\leq\frac{m}{N} \left(\frac{4(d-1)}{d+1}+\frac{2}{d}\right)\\
&
\leq \frac{4d^2}{N}%
\end{aligned}
\end{equation}
Substituting our upper bound of $\sigma^2$ into the matrix Berstein inequality
yields the claim, provided that $\tau$ does not become too large.
\end{proof}

\begin{proposition}
The following bound holds
\textup{\begin{equation}\label{ls_2rob_bound_sc_3}
\mbox{Pr}[\|\hat{\Phi}_{LS}-\Phi\|_2^2\geq\delta^2]\:\leq\:d^2\mbox{exp} \left(-\frac{3N\delta^2}{32 d^4}\right).
\end{equation}}
\end{proposition}
\begin{proof}
As with scenario \eqref{sub:pauli-ancilla}, we use the relation \eqref{frob_op_rltn} to arrive at \eqref{ls_2rob_bound_sc_3}.
\end{proof}

\subsection{Summary of LS concentration bounds}\label{sec:summary_LS}
The LS error bounds derived in sections \ref{sec:ls_bound_sc1},\ref{variance_scenario_2}, \ref{sec:LS_sc3} and \ref{sec:ls_bound_sc4} can be summarised as follows
\begin{corollary}
\label{cor:LS_errors}
\textup{\begin{equation}\label{ls_op_norm_bound_uni}
{\rm Pr}[\|\hat{\Phi}_{LS}-\Phi\|_{\infty}\geq\tau]\:\leq\:d^2\mbox{exp} (-\frac{3N\tau^2}{8}g(k))
\end{equation}}
\textup{\begin{equation}\label{ls_2rob_norm_bound_uni}
{\rm Pr}[\|\hat{\Phi}_{LS}-\Phi\|_2^2\geq\delta^2]\:\leq\:d^2\mbox{exp} (-\frac{3N\delta^2}{8}\frac{g(k)}{d^2})
\end{equation}}
where:
\begin{equation}
\begin{aligned}
g(k)=
\begin{cases}
\frac{1}{3^{2k}
}\:\:\:\:\:\:\:\:\:\:\:\:\:\:\:\:\:\:\:\:\:&\mbox{Scenarios \ref{sub:pauli-ancilla}, \ref{sub:pauli-direct}}\\
\frac{1}{2}\frac{1}{2^{2k}}&\mbox{Scenario \ref{sub:MUB+ancilla}}\\
\frac{1}{4} \frac{1}{2^{2k}}&\mbox{Scenario \ref{sub:MUB-direct}}
\end{cases}
\end{aligned}
\end{equation}
for $\tau$ and $\delta^2\:\in[0,1]$.
\end{corollary}

\section{Key properties for two-step projection}
\label{appndix.proj.prop}
In this section, we prove that the implementations of our projection procedures for computing  $\hat{\Phi}_{PLS}$ satisfy the two properties \eqref{prop_CP1} and \eqref{decrease_dist}.

We 
start with the following lemma relating the operator norm distances of $\Phi_{LS}$ and $\Phi_{CP1}$ from $\Phi$, that is, ensuring Property \eqref{prop_CP1}.
\begin{lemma}\label{cp1_ls_op_norm_ineq}
For any non-negative threshold $\tau$ that is less than $ \|\hat{\Phi}_{LS}-\Phi\|_{\infty}$, the first thresholded CP1 estimator $ \Phi^{\tau}_{CP1}$ defined in Algorithm \ref{proj_CP} satisfies:
\begin{equation}
\|\hat{\Phi}^{\tau}_{CP1}-\Phi\|_{\infty}\leq 2 \|\hat{\Phi}_{LS}-\Phi\|_{\infty}.
\end{equation}
In particular, if $\tau$ is zero (projection on CP) or if $\tau = - \lambda_{\min}(\hat{\Phi}_{LS})$ (our main implementation), the condition is satisfied.
\end{lemma}

\begin{proof}
We write $\lambda_i$ for the eigenvalues of $ \hat{\Phi}_{LS}$, as in Algorithm \ref{proj_CP}.

First, since $\Phi$ is positive semi-definite, all negative eigenvalues of $\hat{\Phi}_{LS} $ have smaller absolute value than $ \|\hat{\Phi}_{LS}-\Phi\|_{\infty}$. In particular $ \tau = - \lambda_{\min}(\hat{\Phi}_{LS})$ satisfies the conditions of the lemma.

The operators $\hat{\Phi}^{\tau}_{CP1}$ and $\hat{\Phi}_{LS}$ are diagonal in the same basis. By the triangle inequality, we just have to prove that no eigenvalue gets changed by more than $ \|\hat{\Phi}_{LS}-\Phi\|_{\infty}$.

If the condition on line \ref{condplus} is met, then all eigenvalues between $\lambda_{\min}(\hat{\Phi}_{LS})$ and $\tau$ are set to zero, and all the other eigenvalues $\mu_i$ are set to $(\lambda_i + \tau - x_0)_+$. So that the small eigenvalues change by at most $\tau \vee |\lambda_{\min}(\hat{\Phi}_{LS})| \leq \|\hat{\Phi}_{LS}-\Phi\|_{\infty}$, and the big eigenvalues change by at most $\tau \vee |\tau - x_0|$. We must then bound $|\tau - x_0|$. First notice that $x_0 \geq 0$ since $\sum \mu_i \geq 1$ and $\sum (\mu_i - x_0)_+ = 1$. Our worst case is then if $x_0 > 2 \tau$. But in that case, all final eigenvalues can be seen as $(\lambda_i + \tau - x_0)_+ $. Since $ \left\langle v_i \middle| \Phi \middle| v_i \right\rangle \geq (\lambda_i - \|\hat{\Phi}_{LS}-\Phi\|_{\infty})_+$, and $\Tr\,\Phi = 1$, we deduce $x_0 - \tau \leq \|\hat{\Phi}_{LS}-\Phi\|_{\infty}$, as wished.

On the other hand, if the condition on line \ref{condplus} is not met,
all eigenvalues bigger than the critical eigenvalue $\lambda_j$ are increased by $\tau$, all smaller eigenvalues are set to zero, and the final eigenvalue $j$ satisfies $0 \leq \mu_{j} \leq \lambda_j + \tau$. So that we only have to prove that $0 \leq \lambda_j \leq \tau$. Since the trace of $\hat{\Phi}_{LS}$ is one, the left inequality on line \ref{findj} ensures that $\lambda_j > 0$. On the other hand, if $\lambda_j > \tau$, then the right inequality on line \ref{findj} would ensure that the the condition on line \ref{condplus} was met. The contradiction ends the proof.

\end{proof}

We now prove that the second `projection' does not increase the Frobenius distance, that is, Property \eqref{decrease_dist}.
\begin{lemma}
\label{lem:decrease_dist}
Assume that $\hat{\Phi}_{PLS}$ is, given the input $\hat{\Phi}_{CP1}$, the output  of either
Dykstra's algorithm at convergence, AP algorithm for any tolerance, or hyperplane intersection projections algorithm with any switch conditions for any tolerance. Then
\begin{equation}
    \|\hat{\Phi}_{PLS}-\Phi\|_{2} \leq   \|\hat{\Phi}_{CP1}-\Phi\|_{2} 
\end{equation}
\end{lemma}
\begin{proof}

By Lemma \ref{convex_projection_lemma}, we merely have to prove that all our `projections' can be seen as a succession of usual projections in Frobenius norm on convex sets that contain $\CPTP$.

Dykstra's algorithm directly outputs the Frobenius projection on $\CPTP$.

AP are a succession of Frobenius projections on $\CP$, and on $\TP$, both of which are convex set that include $\CPTP$.

Hyperplane intersection projections in AP mode is just AP. In HIP mode, we project on the intersection of: $\TP$, and a set of half-spaces that contain $\CP$. All those sets are convex sets that contain $\CPTP$, hence their intersection also is.

\end{proof}

\section{Proof of Theorem \ref{thm:projected_bounds}}\label{sec:projected_bounds}

We now return to the proof of Theorem \ref{thm:projected_bounds} and use the concentration bounds \eqref{ls_op_norm_bound_uni} for the LS estimators in the 4 scenarios to obtain bounds for PLS.

We will prove a stronger version, Theorem \ref{thm:projected_bounds_complete_general} which also covers the proof of Theorem \ref{thm:projected_bounds_general}.

To state it, we recall the notion of being close to rank $r$: a state $\rho$ is $\delta$-almost rank $r$ if there is a rank $r$ state $\rho_{(r)}$ such that
\begin{align}
    \| \rho - \rho_{(r)} \|_{\infty} & \leq \delta.
\end{align}

\begin{theorem}\label{thm:projected_bounds_complete_general}
Let $\Phi$ be the Choi matrix of a channel and $\hat{\Phi}_{PLS}$ our estimator, generated with any `projections' satisfying Properties \eqref{prop_CP1} and \eqref{decrease_dist}.
Then, with probability at least $1 - \eta$, the bounds \eqref{Frob_concl}, \eqref{L1_concl_Phi} and \eqref{L1_concl_PhiCP} below hold simultaneously true, for all $r$ and corresponding values of $\delta$.

Assume that either $\Phi$ or $\hat{\Phi}_{CP1}$ is $\delta$-almost rank $r$. 
With probability at least $1-\eta$:
\begin{equation}\label{Frob_concl}
   \|\hat{\Phi}_{PLS}-\Phi\|_2^{2} \leq 2r \left(\delta + 2\sqrt{\frac{8 \ln(d^2/\eta)}{3 N g(k)}}\right)^2, 
\end{equation}
If it was $\Phi$ that was $\delta$-almost rank $r$, then:
\begin{equation}\label{L1_concl_Phi}
\begin{aligned}
   \|\hat{\Phi}_{PLS}-&\Phi\|_1 \leq r \left((4\sqrt{2} + 2)\delta + 4\sqrt{2}\sqrt{\frac{8 \ln(d^2/\eta)}{3 N g(k)}}\right),
\end{aligned}
\end{equation}
If it was  $\hat{\Phi}_{CP1}$ that was $\delta$-almost rank $r$, then:
\begin{equation}\label{L1_concl_PhiCP}\begin{aligned}
   \|\hat{\Phi}&_{PLS}-\Phi\|_1 \leq  r \left((4\sqrt{2} + 2)\delta + (4 + 8\sqrt{2})\sqrt{\frac{8 \ln(d^2/\eta)}{3 N g(k)}}\right).
\end{aligned}
\end{equation}

\end{theorem}

As a remark, Theorem \ref{thm:projected_bounds_complete_general} reduces to Theorem \ref{thm:projected_bounds} when $\Phi$ is $\delta$-almost rank $r$, with $\delta=0$, i.e. it is exactly rank $r$. The formulation in Theorem \ref{thm:projected_bounds_complete_general} can be seen by taking $\epsilon$ to equal the right hand side of the equations, solving for $\eta$ and inverting the probabilities.

First, Property \eqref{prop_CP1} relates the $L^{\infty}$ norms of the errors of $\hat\Phi_{LS}$ and $\hat\Phi_{CP1}$:
\begin{align*}
\|\hat{\Phi}^{\tau}_{CP1}-\Phi\|_{\infty} & \leq 2 \|\hat{\Phi}_{LS}-\Phi\|_{\infty}.
\end{align*}

The next step 
is to relate the Frobenius norm of the distance between $\hat{\Phi}_{CP1}$ and $\Phi$ to their operator norm distance.

\begin{lemma}\label{PLS_2rob_cp1_inf_ineq}
Assume either that $\hat{\Phi}_{CP1}$ or $\Phi$ is $\delta$-almost rank $r$. Then
\begin{equation}
\|\hat{\Phi}_{CP1}-\Phi\|_2^2\:\leq\:2r(\|\hat{\Phi}_{CP1}-\Phi\|_\infty +  \delta)^2.
\end{equation}
\end{lemma}
\begin{proof}
A key argument is presented in \cite{Gu__2020} for the following inequality relating the 1-norm to the operator-norm of two trace-1 positive-semidefinite matrices $\rho_1$ (of definite rank $r$) and $\rho_2$:
\begin{equation}
\|\rho_1-\rho_2\|_1\:\leq\:2r\|\rho_1-\rho_2\|_\infty.
\end{equation}
Thus we have:
\begin{align}
\|\hat{\Phi}_{CP1}-\Phi\|_1 & \leq \|\hat{\Phi}_{CP1}-\Phi^r\|_1 + \|\Phi^r-\Phi\|_1\\ 
& \leq 2r( \|\hat{\Phi}_{CP1}-\Phi^r\|_{\infty} + \|\Phi^r-\Phi\|_{\infty})\\
& \leq 2r(\delta + \|\hat{\Phi}_{CP1}-\Phi\|_{\infty} + \delta).
\end{align}
There exists the following inequality relating the Frobenius, 1- and operator-norms:\\
Let $A\in\mbox{Herm}_{n\times n}$ be an arbitrary Hermitian matrix, then
\begin{equation}
\|A\|_2^{2}\:\leq\:\|A\|_{\infty}\|A\|_{1},
\end{equation}
which we can make use of since $\hat{\Phi}_{CP1}-\Phi$ is the difference of two quantum states, and hence must be Hermitian.\par
Using this, we can say:
\begin{align*}
\|\hat{\Phi}_{CP1}-\Phi\|_2^2 & \leq 2r(\|\hat{\Phi}_{CP1}-\Phi\|_{\infty} + 2\delta)\|\hat{\Phi}_{CP1}-\Phi\|_\infty \\
& \leq 2r(\|\hat{\Phi}_{CP1}-\Phi\|_\infty +  \delta)^2.
\end{align*}

\end{proof}

By Property~\ref{decrease_dist} we may transfer this result to $\hat{\Phi}_{PLS}$:
\begin{equation}
\|\hat{\Phi}_{PLS}-\Phi\|_2^2\:\leq\:2r(\|\hat{\Phi}_{CP1}-\Phi\|_\infty +  \delta)^2.
\end{equation}

We can now pull these results together to prove Theorem  \ref{thm:projected_bounds_complete_general}. 
\\\par 
If $\|\hat{\Phi}_{LS}-\Phi_{\infty}\|\:\leq\:\tau$ with some probability, then  $\|\hat{\Phi}_{CP1}-\Phi\|_{\infty}\:\leq\:2\tau$ with at least that probability by Lemma \ref{cp1_ls_op_norm_ineq}. Then, by Lemma \ref{PLS_2rob_cp1_inf_ineq}
\begin{equation}
\|\hat{\Phi}_{PLS}-\Phi\|_2^2\:\leq\:2r(2\tau + \delta)^2
\end{equation}
with yet greater probability.\par

Using Equation \eqref{ls_op_norm_bound_uni}, and setting the right-hand-side equal to $\eta$, we get $\tau = \sqrt{\frac{8 \ln(d^2/\eta)}{3 N g(k)}}$. We thus obtain, with probability at least $1 - \eta$:
\begin{equation}
   \|\hat{\Phi}_{PLS}-\Phi\|_2^{2} \leq 2r \left(\delta + 2\sqrt{\frac{8 \ln(d^2/\eta)}{3 N g(k)}}\right)^2, 
\end{equation}
where, to remind the reader:
\begin{equation}
\begin{aligned}
g(k)=
\begin{cases}
\frac{1}{3^{2k}}\:\:\:\:\:\:\:\:\:\:\:\:\:\:\:\:\:\:\:\:\:&\mbox{Scenarios \ref{sub:pauli-ancilla}, \ref{sub:pauli-direct}}\\
\frac{1}{2}\frac{1}{2^{2k}}&\mbox{Scenario \ref{sub:MUB+ancilla}}\\
\frac{1}{4}\frac{1}{2^{2k}}&\mbox{Scenario \ref{sub:MUB-direct}}
\end{cases}
\end{aligned}
\end{equation}
Recall that Equation \eqref{ls_op_norm_bound_uni} is valid if  $\tau\in[0,1]$. But if $\tau$ is more than one, then the right-hand side of Equation \ref{Frob_concl} is more than two, hence always true. The expression is thus always valid.

\vspace{4mm}

For the norm-one bound we use inequality \eqref{eq.norm-one.frobenius.low.rank} below. 
Let $\rho$ and $\rho_{(r)}$ be states, and assume that $\rho$ is rank $r$. We write 
$
\varphi = \rho_{(r)} - \rho
$
and consider its decomposition into positive and negative parts
$
\varphi = \varphi_+ - \varphi_-
$
Since $\rho_{(r)}$ has rank $r$, the positive 
part $\varphi_+$ has at most rank $r$. Moreover, since ${\rm Tr}(\varphi) =0$, we have that
\begin{eqnarray}\label{eq.norm-one.frobenius.low.rank}
\|\rho_{(r)} - \rho\|_1 &=& 
2{\rm Tr}(\varphi_+) \leq 2\sqrt{r} \|\varphi_+\|_2 \\
&\leq& 2\sqrt{r}
\|\rho_{(r)} - \rho\|_2.
\end{eqnarray}

With this inequality, we may write:
\begin{align*}
    & \phantom{\leq} \|\hat{\Phi}_{PLS} - \Phi\|_1 \\
    & \leq  \|\hat{\Phi}_{PLS} - \Phi^r\|_1 +  \|\Phi - \Phi^r\|_1 \\
    & \leq 2\sqrt{r} \left(  \|\hat{\Phi}_{PLS} - \Phi^r\|_2 \right) +  \|\Phi - \Phi^r\|_1 \\
    & \leq 2\sqrt{r} \left(  \|\hat{\Phi}_{PLS} - \Phi\|_2 +  \|\Phi - \Phi^r\|_2 \right) +  \|\Phi - \Phi^r\|_1.
\end{align*}
Now, if $\Phi^r$ was the approximation to the $\delta$-almost rank $r$ true state $\Phi$, then 
$\| \Phi- \Phi^r\|_1 \leq 2r \delta$ and $ \| \Phi- \Phi^r\|_2 \leq \sqrt{2r} \delta$. If $\Phi^r$ was  the approximation to the $\delta$-almost rank $r$ CP projection $\hat{\Phi}_{CP1}$, then 
$\| \Phi- \Phi^r\|_1 \leq 2r (\delta + \| \Phi - \hat{\Phi}_{CP1}\|_{\infty}) $ and $ \| \Phi- \Phi^r\|_2 \leq \sqrt{2r} (\delta + \| \Phi - \hat{\Phi}_{CP1}\|_{\infty} )$. 

Substituting with the upper bound $\tau$ and using Inequality  \eqref{eq.norm-one.frobenius.low.rank} ends the proof.
\qed



\section{Error bounds and sampling complexities for the direct projection method}\label{app:direct_projection_results}




\begin{theorem}\label{thm:direct_projected_bounds}
The squared Frobenius norm errors of the estimators $\hat{\Phi}_{PLS}$ of a Choi matrix $\Phi$ representing a $k$-qubit channel, found by projecting the LS estimators scenarios \ref{sub:pauli-ancilla}, \ref{sub:pauli-direct},  \ref{sub:MUB+ancilla} and \ref{sub:MUB-direct}, \textit{directly} onto $\C\mathcal{P}\T\mathcal{P}$ satisfy the following bound:
\textup{\begin{equation}\label{direct_condensed_error_bound}
\mbox{Pr}[\|\hat{\Phi}_{PLS}-\Phi\|_{2}^{2}\geq\epsilon]\:\leq\:d^2\mbox{exp}\left(-\frac{3N\epsilon}{8}f(k)\right).
\end{equation}}
where:
\textup{\begin{equation}
\begin{aligned}
f(k)=
\begin{cases}
\frac{1}{3^{2k}\times2^{2k}}\:\:\:\:\:\:\:\:\:\:\:\:\:\:\:\:\:\:\:\:\:\:\:\:\:\:\:\:&\mbox{Scenarios \ref{sub:pauli-ancilla}, \ref{sub:pauli-direct}}\\
\frac{1}{2}\frac{1}{2^{4k}}&\mbox{Scenario \ref{sub:MUB+ancilla}}\\
\frac{1}{4}\frac{1}{2^{4k}}&\mbox{Scenario \ref{sub:MUB-direct}}\\
\end{cases}
\end{aligned}
\end{equation}}
for \textup{$\epsilon\in [0,1]$}.
\end{theorem}\par

\begin{proof}

Using Lemma \ref{convex_projection_lemma}, one can take the bound defined in \eqref{ls_2rob_norm_bound_uni} and simply apply it to $\hat{\Phi}_{PLS}$ to arrive at equation \eqref{direct_condensed_error_bound}.
\end{proof}

Inverting Theorem \ref{thm:direct_projected_bounds} we prove following:
\begin{corollary}\label{thm:direct_sampling_complextites}
To achieve the following accuracy:
\textup{\begin{equation}
\mbox{Pr}[\|\hat{\Phi}_{PLS}-\Phi\|_2^2\:\geq\epsilon\:]\:\leq\:\eta
\end{equation}}
one requires a sample size of
\textup{\begin{equation}\label{direct_sample sizes}
N\geq\frac{1}{f(k)}\frac{8}{3\epsilon}\mbox{log}(\frac{d^2}{\eta}).
\end{equation}}
The sampling complexities for each method of data collection, with direct projection, are then:
\textup{\begin{equation}
\begin{aligned}
N(k)=
\begin{cases}
\bigo(\frac{1}{\epsilon}3^{2k}\times2^{2k})\approx \bigo(\frac{1}{\epsilon}d^{5.17})\:\:\:\:&\mbox{Scenarios 1/2}\\
\bigo(\frac{1}{\epsilon}d^4)&\mbox{Scenarios 3/4}\\
\end{cases}
\end{aligned}
\end{equation}}
up to logarithmic factors.
\end{corollary}

Equation \eqref{condensed_error_bound} is to be contrasted with \eqref{direct_condensed_error_bound}. Clearly, the bound \eqref{condensed_error_bound} is stricter than that of equation \eqref{direct_condensed_error_bound} when $r<\frac{d^2}{32}$. Therefore, the first method of projecting $\hat{\Phi}_{CP1}$ onto $\C\mathcal{P}\T\mathcal{P}$ improves upon the error bound when the channel acts on more than 5 qubits, and for certain ranks in states of 3 to 5 qubits, while for 1 and 2 qubits the direct projection bound is tighter.
Most importantly, the bound for the two steps projection method captures the qualitative dependence on the rank of the channels, which is confirmed by the numerical simulations described in Section \ref{subsec:rank.numerics}.\\\par


\section{Expressions for Direct Projections}\label{app:details_on_num_expts}

\subsection{Solution to the TP Projection}\label{sec:tp_projection}
The projection, with respect to the Frobenius norm, of a matrix $X$ onto the set of matrices representing trace-preserving maps is the solution to the following optimisation problem:
\begin{equation}\label{projTP}
\begin{aligned}
\mbox{Proj}_{\mathcal{TP}}[X] &= \underset{X'}{\mbox{argmin}}\|X-X'\|^{2}_{2} \\
&\quad \mbox{s.t. $\mbox{Tr}_s(X')=\frac{1}{d}\mathbb{1}_d$}
\end{aligned}
\end{equation}.

\begin{proposition}[Closed-form expression for $\mathrm{TP}$-projection] \label{tp_proj_proposition}
The unique solution to the optimization problem~\eqref{projTP} is
\begin{align*}
\mbox{Proj}_{\mathcal{TP}}[X] = X + \tfrac{1}{d}\mathbb{1} \otimes \left( \tfrac{1}{d} \mathbb{1} - \mathrm{tr}_s (X) \right).
\end{align*}
\end{proposition}

\begin{proof}
Let us start by defining the projection onto the linear subspace of trace annihilating maps $\mathcal{TA} = \left\{X:\; \mathrm{tr}_s (X) =0 \right\}$:
\begin{align*}
\mbox{Proj}_{\mathcal{TA}}[X] = X - \tfrac{1}{d} \mathbb{1} \otimes \mathrm{tr}_s (X).
\end{align*} 
It is easy to check that this operator defines a projection with respect to the Frobenius norm ($\mbox{Proj}_{\mathcal{TA}}\big[\mbox{Proj}_{\mathcal{TA}}[X]\big]=\mbox{Proj}_{\mathcal{TA}}[X]$ and 
$\mathrm{tr} \big( \mbox{Proj}_{\mathcal{TA}}[X] Y \big) = \mathrm{tr} \big( X \mbox{Proj}_{\mathcal{TA}}[Y]\big)$ for all matrices $X,Y$ with compatible dimension)
and we denote its ortho-complement by 
\begin{equation*}
\mbox{Proj}_{\mathcal{TA}}^\bot [X] = X - \mbox{Proj}_{\mathcal{TA}}[X] = \tfrac{1}{d} \mathbb{1} \otimes \mathrm{tr}_s (X).
\end{equation*}
For the matrix $X$ we introduce the short-hand notation $X_{\mathcal{TA}} = \mbox{Proj}_{\mathcal{TA}}[X]$ and $X_{\mathcal{TA}}^\bot = \mbox{Proj}_{\mathcal{TA}}^\bot [X]$. 
The following identity is essentially the Pythagorean theorem:
\begin{align}
\|X \|_2^2 =& \|X_{\mathcal{TA}}+ X_{\mathcal{TA}}^\bot \|_2^2 = \|X_{\mathcal{TA}} \|_2^2 + \|X_{\mathcal{TA}}^\bot \|_2^2. \label{eq:pythagoras}
\end{align}
Applying it to $X-X'$ with $X' \in \mathcal{TA}$ ($\mathrm{tr}_s (X') = \tfrac{1}{d}\mathbb{1}$) yields
\begin{align*}
\|X-X' \|_2^2 =& \| (X-X')_{\mathcal{TA}}\|_2^2 + \| (A-Y)_{\mathcal{TA}}^\bot \|_2^2 \\
=& \left\| X + \tfrac{1}{d} \mathbb{1} \otimes \left( \tfrac{1}{d} \mathbb{1} - \mathrm{tr}_s (X) \right) - X' \right\|_2^2 \\
+& \left\| \tfrac{1}{d} \mathbb{1} \otimes \left( \tfrac{1}{d} \mathbb{1} - \mathrm{tr}_s (X) \right) \right\|_2^2.
\end{align*}
The second term is independent from the actual choice of $X'$. The first term is minimal for
\begin{equation*}
X'_\sharp = X + \tfrac{1}{d} \mathbb{1} \otimes \left( \tfrac{1}{d} \mathbb{1}-\mathrm{tr}_1 (X) \right) \in \mathcal{TP}.
\end{equation*}
\end{proof}

The argument above conveys another interesting piece of information. Projecting onto the set of trace-preserving maps decreases the Frobenius distance. Suppose that $Y \in \mathcal{TP}$ and $X$ is arbitrary.
Then,
\begin{align*}
\left\| \mbox{Proj}_{\mathcal{TP}}[X] - Y \right\|_2^2 =& \left\| X- Y \right\|_2^2 - \tfrac{1}{d} \left\| \tfrac{1}{d} \mathbb{1} - \mathrm{tr}_s (X) \right\|_2^2.
\end{align*}
The right hand side is equal to $\|X-Y \|_2^2$ if and only if $X$ is also trace preserving. Otherwise it is strictly smaller.

\subsection{Solution to the CP Projection}\label{sec:cp_projection}

The projection of a matrix $X$ onto the set of positive-semidefinite matrices is the solution to the following optimisation problem:
\begin{equation}\label{projCP}
\begin{aligned}
\mbox{Proj}_{\mathcal{CP}}[X] &= \underset{X'}{\mbox{argmin}}\|X-X'\|^{2}_{2} \\
&\quad \mbox{s.t. $X'$ positive semi-definite}
\end{aligned}
\end{equation}

The condition of positive semidefiniteness is that all eigenvalues be greater than or equal to zero. An obvious method, therefore, for enforcing the positive semidefiniteness of a matrix is to set all negative eigenvalues to zero. This turns out to be the unique solution to equation \eqref{projCP} \cite{Smolin_2012}. We provide a short proof for the sake of being self-contained.

\begin{proposition}[Closed-form expression for $\mathrm{PSD}$-projection]
Fix a self-adjoint matrix $X$ with eigenvalue decomposition $X = \sum_i \xi_i |x_i \rangle \! \langle x_i |$.  Then, the unique solution to optimization problem~\eqref{projCP} is 
\begin{equation}
\mbox{Proj}_{\mathcal{CP}}[X] = X_+= \sum_{i} \max \left\{0,\xi_i \right\} |x_i \rangle \! \langle x_i| \label{cp_proj_sltn}
\end{equation} 
\end{proposition}

\begin{proof}
Decompose $X = X_+-X_-$ such that  $X_+,X_-$ are both psd and orthogonal with respect to the Frobenius inner product ($\mathrm{tr}(X_+ X_-) = 0$).
For $X'$ positive semidefinite, we can expand the squared Frobenius norm difference as
\begin{align*}
\|X-X' \|_2^2 =& \| (X_+ - X') - X_- \|_2^2  \\
=& \| X_+ - X' \|_2^2 - 2 \mathrm{tr}((X_+-X')X_-) + \|X_- \|_2^2 \\
=& \|X_+ - X' \|_2^2 + 2 \mathrm{tr}(X_- X' ) + \|X_- \|_2^2,
\end{align*}
where we have used orthogonality between $X_+$ and $X_-$.
Note that $\mathrm{tr}(X_-X') \geq 0$, because both $X_-$ and $X'$ are psd and $\|X_+ -X' \|_2^2 \geq 0$. 
Hence,
\begin{align*}
\|X - X' \|_2^2 \geq \|X_- \|_2^2 \quad \text{for all $X'$ positive semidefinite.}
\end{align*}
Equality holds if and only if $X = X_+$.
\end{proof}

This argument also conveys another interesting piece of information: Projecting onto the set of positive-semidefinite matrices decreases the squared Frobenius distance. Suppose that $Y$ is positive semidefinite and $X$ is arbitrary. Then,
\begin{align*}
\left\| \mbox{Proj}_{\mathcal{CP}}[X]- Y \right\|_2^2
\leq \left\| X - Y \right\|_2^2 - \left\| X - \mbox{Proj}_{\mathcal{CP}}[X] \right\|_2^2
\end{align*}
and we can always compute the minimal shrinkage $\left\| X - \mbox{Proj}_{\mathcal{CP}}[X] \right\|_2^2$ explictly.
With the projection defined as in \eqref{cp_proj_sltn}, we note that the computational cost of projecting onto the the CP cone comes almost entirely from that of computing the eigen-decomposition of the estimator.

\subsection{Solution to the CP1 Projection}\label{sec:cp1_proj_sltn}
The final projection we discuss is that on trace one positive matrices, which is required in the first step of the PLS method.
The projection itself is similar to the CP projection, the difference being that the output obeys the additional constraint 
\begin{equation}
\mbox{Tr}(\:\mbox{Proj}_{CP1}(\hat{\Phi}_{LS})\:)\overset{!}{=}1
\end{equation}
Define $\vec{\lambda'}=\vec{\lambda}-\vec{x_0}$, where $\vec{x_0}=(x_0,\cdots,x_0)$ is a vector that will ensure unit trace. Then the solution is a simple modification of the CP projection \cite{Smolin_2012}:
\begin{equation}\label{cp1_proj_sltn}
\mbox{Proj}_{CP1}[\hat{\Phi}_{LS}]\: =\: \mathcal{V}\mbox{diag}(\vec{\lambda'}_+)\mathcal{V}^{\dagger}
\end{equation}
where $\vec{\lambda'}_+=(\mbox{max}(\lambda_1',0),\cdots,\mbox{max}(\lambda_{d^2}',0))$. The parameter $x_0$ can be found by solving the equation:
\begin{equation}\label{finding_x0}
2+\mbox{Tr}(\hat{\Phi}_{LS})-d^2x_0+\sum_{i=1}^{d^2}|\lambda_i-x_0|\:=\:0.
\end{equation}
Such constrained projection is well-known, see for instance \cite{Smolin_2012,Gu__2020}. (The result appears to be have been known earlier in the optimization community but we were unable to find a reference.)\\

\section{Additional Projection Algorithms}\label{sec:additional_proj_algs}
The algorithms HIPswitch, OneHIP and PureHIP explored experimentally in Section \ref{sec:algo_comp} are based on Algorithm \ref{HIP}. Here, we show explicitly their individual implementations. We rely on Algorithm \ref{HIP_intern}.\par
First, the simplest implementation: OneHIP, in which we carry out no AP iterations, and only project onto the single, most recent, hyperplane (Algorithm \ref{alg:oneHIP}).

\begin{algorithm}[!ht]
\caption{oneHIP}\label{alg:oneHIP}
\begin{algorithmic}[1]
\Function{oneHIP}{$\Phi$ initial state, $\epsilon$ tolerance}
\While{$\lambda_{\min}(\Phi) < -\epsilon$}
\State $\texttt{w\_active} \gets \textrm{empty list}$
\State $\Phi_{\CP} \gets \texttt{proj}_{\CP}(\Phi)$
\State $\texttt{w} \gets $ \textrm{half-space containing $\CP$ orthogonal at $\Phi_{\CP}$ to $\Phi_{\CP} - \Phi$}
\State Add $\texttt{w}$ as first element of $\texttt{w\_active}$
\State $\texttt{w\_active}, \Phi \gets \texttt{HIP\_inner}(\texttt{w\_active}, \Phi)$
\EndWhile
\State \Return $\Phi$
\EndFunction
\end{algorithmic}
\end{algorithm}

Next, PureHIP, in which still do no AP iterations, and we keep a number of hyperplanes in memory (the list \texttt{w\_active} in the algorithm below), and project onto their intersection (Algorithm \ref{alg:pureHIP}).

\begin{algorithm}[!ht]
\caption{pureHIP}\label{alg:pureHIP}
\begin{algorithmic}[1]
\Function{pureHIP}{$\Phi$ initial state, $\epsilon$ tolerance}
\While{$\lambda_{\min}(\Phi) < -\epsilon$}
\State $\Phi_{\CP} \gets \texttt{proj}_{\CP}(\Phi)$
\State $\texttt{w} \gets $ \textrm{half-space containing $\CP$ orthogonal at $\Phi_{\CP}$ to $\Phi_{\CP} - \Phi$}
\State Add $\texttt{w}$ as first element of $\texttt{w\_active}$
\State $\texttt{w\_active}, \Phi \gets \texttt{HIP\_inner}(\texttt{w\_active}, \Phi)$
\EndWhile
\State \Return $\Phi$
\EndFunction
\end{algorithmic}
\end{algorithm}

Finally, we show HIPswitch with the condition that we do 6 iteration of AP, followed by 30 iterations of HIP (Algorithm \ref{alg:HIPswitch})

\begin{algorithm}[!ht]
\caption{HIPswitch}\label{alg:HIPswitch}
\begin{algorithmic}[1]
\Function{HIPswitch}{$\Phi$ initial state, $\epsilon$ tolerance}
\State SwitchCondition\_toHIP: i is divisible by 6
\State SwitchCondition\_toAP: j is divisible by 30
\State \texttt{i} $=1$
\State \texttt{j} $=1$
\State \texttt{mode} $\gets$ \texttt{AP}
\While{$\lambda_{\min}(\Phi) < -\epsilon$}
\If{\texttt{mode} $=$ \texttt{AP}}
\State $\Phi \gets \texttt{proj}_{\CP}(\Phi)$
\State $\Phi \gets \texttt{proj}_{\TP}(\Phi)$
\State \texttt{i} = \texttt{i} + 1
\If{SwitchCondition\_toHIP}
\State $\texttt{mode}\gets\texttt{HIP}$
\State $\texttt{w\_active} \gets \textrm{empty list}$
\EndIf
\EndIf
\If{\texttt{mode} $=$ \texttt{HIP}}
\State $\Phi_{\CP} \gets \texttt{proj}_{\CP}(\Phi)$
\State $\texttt{w} \gets $ \textrm{half-space containing $\CP$ orthogonal at $\Phi_{\CP}$ to $\Phi_{\CP} - \Phi$}
\State Add $\texttt{w}$ as first element of $\texttt{w\_active}$
\State $\texttt{w\_active}, \Phi \gets \texttt{HIP\_inner}(\texttt{w\_active}, \Phi)$
\State \texttt{j} = \texttt{j} + 1
\If{SwitchCondition\_toAP}
\State $\texttt{mode}\gets\texttt{AP}$
\EndIf
\EndIf
\EndWhile
\State \Return $\Phi$
\EndFunction
\end{algorithmic}
\end{algorithm}



\end{document}